\documentclass[nofootinbib, a4paper, aps, pra, reprint]{revtex4-1}

\usepackage[T1]{fontenc}
\usepackage{amsfonts, amsmath, amssymb, amsthm} 
\usepackage{xfrac, dsfont, braket} 
\usepackage{hyperref, xcolor, appendix, titlesec} 

\definecolor{darkred}{RGB}{150, 0, 0}
\definecolor{darkgreen}{RGB}{0, 100, 0}
\definecolor{darkblue}{RGB}{0, 0, 200}
\hypersetup{colorlinks, breaklinks, linkcolor = darkred, urlcolor = darkblue, citecolor = darkgreen} 

\newtheorem{lemma}{Lemma} 
\newtheorem{theorem}{Theorem}
\newtheorem{corollary}{Corollary}

\titleformat{\section}{\small\bfseries\uppercase}{\thesection.}{0.7em}{\centering} 

\DeclareMathOperator{\tr}{tr} 

\begin{document}

\title{Optimality of any pair of incompatible rank-one projective measurements for some non-trivial Bell inequality}

\author{Gabriel Pereira Alves}
\email{gpereira@fuw.edu.pl}
\affiliation{Faculty of Physics, University of Warsaw, Pasteura 5, 02-093 Warsaw, Poland}

\author{J\k{e}drzej Kaniewski}
\email{jkaniewski@fuw.edu.pl}
\affiliation{Faculty of Physics, University of Warsaw, Pasteura 5, 02-093 Warsaw, Poland}

\date{\today}

\begin{abstract}
Bell non-locality represents one of the most striking departures of quantum mechanics from classical physics. It shows that correlations between space-like separated systems allowed by quantum mechanics are stronger than those present in any classical theory. In a recent work [Sci. Adv. 7, eabc3847 (2021)], a family of Bell functionals tailored to mutually unbiased bases (MUBs) is proposed. For these functionals, the maximal quantum violation is achieved if the two measurements performed by one of the parties are constructed out of MUBs of a fixed dimension. Here, we generalize this construction to an arbitrary incompatible pair of rank-one projective measurements. By constructing a new family of Bell functionals, we show that for any such pair there exists a Bell inequality that is \emph{maximally} violated by this pair. Moreover, when investigating the robustness of these violations to noise, we demonstrate that the realization which is most robust to noise is not generated by MUBs.
\end{abstract}

\maketitle

\section{Introduction}

\noindent The discovery of Bell non-locality lies among the most fundamental results of twentieth-century physics. While the theoretical description was proposed by Bell in 1964 \cite{Bell64}, the first experimental demonstration was performed by Freedman and Clauser \cite{Freedman72}, followed by the seminal work of Aspect et al.~\cite{Aspect81, Aspect82a, Aspect82b}. In 2015, three independent groups performed the Bell experiment in a loophole-free manner~\cite{Hensen15, Giustina15, Shalm15}, meaning that they managed to eliminate several issues that could falsify the conclusions of the experiment.

In a nutshell, Bell non-locality states that correlations between spatially separated parties allowed by quantum theory are stronger than those allowed by local-realistic theories \cite{Bell64, Brunner14}. It is easy to see that entanglement \cite{Horodecki09} and incompatibility of measurements \cite{Heinosaari2016} constitute two necessary resources to generate non-local correlations \cite{Fine82}. Conversely, conditions sufficient to generate non-locality are known only for a restricted class of states or measurements. For a bipartite Bell scenario, all pure entangled states can generate non-locality, a statement referred to as Gisin's theorem \cite{Gisin91}. A result of similar generality for incompatibility of measurements was obtained in Ref.~\cite{Wolf09}, in which it is proved that every incompatible pair of projective measurements enables the violation of a Bell inequality. Here, we continue the study of how useful a pair of measurements is for the purpose of generating non-locality. However, our focus is not merely on observing a Bell violation, but on producing the maximal violation allowed by quantum mechanics.

In this work, we focus on pairs of rank-one projective measurements and generalize a framework originally presented in Ref.~\cite{Tavakoli21} in which a family of bipartite Bell functionals is tailored to mutually unbiased bases (MUBs) \cite{Schwinger60, Durt10}. In other words, in Ref.~\citep{Tavakoli21} the maximum value obtainable by any quantum realization of the functionals -- here referred to as quantum value and denoted by $\beta_Q$ -- can be achieved if one of the parties implements a pair of rank-one projective measurements with uniform overlaps. Here, we study the same Bell scenario (shown in Fig.~\ref{scenario}) as in Ref.~\cite{Tavakoli21}, parametrized by an integer $d\ge 2$. However, our new functionals are tailored to a more general pair of rank-one projective measurements. More specifically, the measurement operators are constructed out of a pair of orthonormal bases $\{\ket{e_j}\}_{j=1}^{d}$ and $\{\ket{f_k}\}_{k=1}^{d}$ on $\mathbb{C}^{d}$ and let us denote the resulting overlap matrix by
\begin{align}
O_{jk} := |\braket{e_j|f_k}|. \label{overlap}
\end{align}
The only assumption we make is that these measurements are incompatible since this is a necessary condition for non-locality.

The family of Bell functionals introduced in this work is designed so that both the quantum value and a realization that achieves it can be written down explicitly. We also demonstrate that our functionals are non-trivial, meaning that the maximum value achievable by any classical realization -- referred to as local value and denoted by $\beta_L$ -- is strictly smaller than the quantum value. This result is obtained by deriving a non-trivial certification statement for the measurements performed by one of the parties. Finally, for a wide class of functionals, we obtain a device-independent certification of the maximally entangled state of local dimension $d$.

To make our work relevant for experiments, we investigate the robustness to noise of the proposed optimal realizations. In the standard model of noise, in which the state is replaced by an isotropic state, it can be shown that the bigger the gap between $\beta_Q$ and $\beta_L$, the more robust to noise the optimal realization is. For even $d$, we found that the largest gap is achieved if and only if the original measurements correspond to a direct sum of qubit MUBs, i.e.~all the non-zero elements of the matrix $O$ must be equal to $\sfrac{1}{\sqrt{2}}$. This contradicts our initial guess that the most noise-robust realization (in a fixed dimension) would correspond to the one generated by MUBs. For odd $d$, we could not determine the largest possible gap, but we have derived an upper bound on $\beta_Q - \beta_L$ which is valid for all $d \ge 3$. Furthermore, we show analytically that for every odd $d$ there exist measurements that give rise to a realization that is more robust to noise than the realization obtained from MUBs.

\begin{figure}[t!]
\hspace{7cm}
\includegraphics[scale = 1]{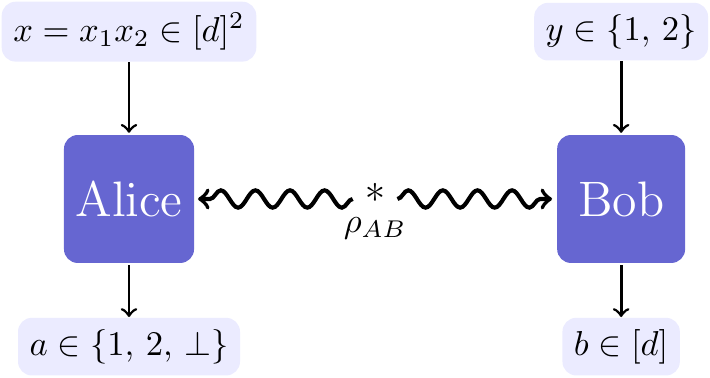}
\caption{A Bell scenario parametrized by $d\ge 2$. The inputs and outputs belong to the following sets: $b,\, x_1,\, x_2\in \{1,\, ...,\, d\}$, $y\in \{1,\, 2\}$ and $a\in \{1,\, 2,\, \perp\}$.} \label{scenario}
\end{figure}

\section{A family of Bell functionals}

\noindent Our goal is to show that every pair of incompatible rank-one projective measurements is optimal for some Bell functional. To achieve this, it is convenient to first apply a simple preprocessing, which simply discards the subspace in which the measurements are compatible (which is not useful for the purpose of generating non-locality). Then, we introduce a bipartite Bell scenario parametrized by an integer $d$ and construct a family of functionals whose quantum value can be computed analytically. Lastly, we present a realization that saturates the quantum value.

Consider a pair of orthonormal bases on $\mathbb{C}^{d^\prime}$, which we denote by $\{\ket{e_{j}}\}_{j=1}^{d^\prime}$ and $\{\ket{f_{k}}\}_{k=1}^{d^\prime}$, and let $O_{jk}^\prime := |\braket{e_{j}|f_{k}}|$ be the overlap matrix. If the two bases share a vector this results in a one-dimensional subspace of $\mathbb{C}^{d^\prime}$ in which the measurements are compatible. We remove all such subspaces by truncating the original Hilbert space appropriately, which leads to a pair of rank-one projective measurements acting on $\mathbb{C}^{d}$ for some $d \leq d^\prime$ (since we assume that the original measurements are incompatible, we are guaranteed that $d \geq 2$). For these new measurements, the overlap matrix is guaranteed to satisfy $O_{jk} < 1$. Moreover, these new measurements can be implemented by performing the original measurements on a quantum state with appropriately chosen local support. From now, we will assume that this process has already been performed and we will restrict our attention to measurements acting on effective dimension $d$ whose overlaps satisfy $O_{jk} < 1$.

Now, consider a bipartite Bell scenario characterized by an integer $d \geq 2$, whose parties are named Alice and Bob. For each $d$, Alice is given a two-character string denoted by $x:=x_1 x_2$, for $x_1,\, x_2\in \{1,\, ...,\, d\}$ and outputs $a\in\{1,\, 2,\, \perp\}$. Bob has two possible inputs labeled by $y\in\{1,\, 2\}$ and outputs $b\in\{1,\, ...,\, d\}$ (see Fig.~\ref{scenario}). We are given a pair of orthonormal bases $\{\ket{e_j}\}_{j=1}^{d}$ and $\{\ket{f_k}\}_{k=1}^{d}$ whose overlap matrix is given by $O_{jk} := |\braket{e_j|f_k}|$ and satisfies $O_{jk} < 1$. We define the \emph{correlation} score $\mathcal{C}_d$ as
\begin{multline}
\mathcal{C}_d : = \sum_{x_1,x_2=1}^d\, \sum_{y=1}^2 \lambda_x \left[\, p( a =y,\, b=x_y|x,\, y) \right.  \\
- \left. p( a = \bar{y},\, b=x_y|x,\, y) \right] , \label{score}
\end{multline}
where $\lambda_x := \sqrt{1-O_{x_1x_2}^2}$ and the notation $\bar{y}$ means that the value of $y$ is flipped from 1 to 2 or vice-versa. Note that the weights $\lambda_{x}$ depend on the overlap matrix and, as we will later see, this ensures that the functional is tailored to the orthonormal bases $\{\ket{e_j}\}_{j=1}^{d}$ and $\{\ket{f_k}\}_{k=1}^{d}$. The only terms that produce a non-zero contribution to $\mathcal{C}_d$ correspond to the cases when $b=x_y$ and $a\neq\perp$. In those cases, the score is increased (decreased) if Alice outputs $a=y$ ($a=\bar{y}$). In order to balance the cases where $a\in\{1,\, 2\}$ (Alice plays the game) and $a=\perp$ (Alice does not play the game), a penalty is introduced\footnote{In the end  of Appendix \ref{proof lemma 1} there is a clarification of why this penalty is necessary.}, giving rise to our \emph{final} score:
\begin{align}
\mathcal{F}_d : =\mathcal{C}_d-\frac{1}{2}\sum_{x_1,x_2=1}^d \lambda_x^2\, \left[\,p(a=1|x)+p(a=2|x)\right]. \label{functional}
\end{align}
It turns out that the quantum value of $\mathcal{F}_d$ does not depend on the overlap matrix $O$, as shown in the lemma below.

\begin{lemma}
For $d\ge 2$, the quantum value of $\mathcal{F}_d$ is $d-1$. \label{lemma1}
\end{lemma}\vspace{-3mm}

\noindent The proof of Lemma \ref{lemma1} consists of finding an upper bound on the value of $\mathcal{F}_d$ and demonstrating that it can be saturated. Proving the upper bound relies on many elementary steps, which we present in Appendix \ref{proof lemma 1}. For our purposes, it suffices to show how to construct a quantum realization that saturates the quantum value of $\mathcal{F}_d$.

We choose the measurements of Bob corresponding to inputs $y = 1$ and $y = 2$ to be $P_{x_1}=\ket{e_{x_1}}\!\!\bra{e_{x_1}}$ and $Q_{x_2}=\ket{f_{x_2}}\!\!\bra{f_{x_2}}$, respectively.\footnote{We use $x_{1}$ ($x_{2}$) as the subscript for $P$ ($Q$) because we only get a non-zero contribution to the functional when $b = x_{y}$.} For the measurements of Alice, it is more convenient to work with Hermitian observables rather than measurement operators. For every input $x$ of Alice, we express the three measurement operators $A_x^{(1)}$, $A_x^{(2)}$ and $A_x^{(\perp)}$ in a compact manner by defining the observable $A_{x}$ as
\begin{multline}
A_x: = \,(+1)\times A_x^{(1)}+(-1)\times A_x^{(2)}+0\times A_x^{(\perp)} \\ = A_x^{(1)}-A_x^{(2)}. \label{alicesobs}
\end{multline}
It is easy to see that if the measurement is projective, then the spectrum of $A_x$ belongs to the set $\{\pm 1,\, 0\}$. In addition, any Hermitian operator whose spectrum is contained in this set can be interpreted as an observable arising from a projective measurement, so the two representations are equivalent. Then, let us choose $A_x$ as:
\begin{align}
A_x= \frac{1}{\lambda_x}\left(P_{x_1}-Q_{x_2}\right)^\textsf{T}, \label{alicesmeas}
\end{align}
where $(\cdot )^\textsf{T}$ denotes the transposition in the computational basis. Because $P_{x_1}$ and $Q_{x_2}$ are rank-one, the difference $P_{x_1}-Q_{x_2}$ is a rank-two operator. With a simple calculation, it is possible to check that the spectrum of $P_{x_1}-Q_{x_2}$ is contained in $\{\pm \lambda_x,~0\}$. Thus, since $O_{x_1x_2}<1$ and $\lambda_x = \sqrt{1 - O_{x_1x_2}^2}$, we have $\lambda_x>0$, and the normalising factor ensures the spectrum of $A_x$ is contained in $\{\pm 1,\, 0\}$.\footnote{Recall that taking the transposition does not affect the spectrum.}

Finally, we pick the shared state as the maximally entangled state belonging to $\mathbb{C}^d\otimes\mathbb{C}^d$:
\begin{align}
\ket{\Phi_d^+}=\frac{1}{\sqrt{d}}\sum_{k=1}^d\ket{k,\, k}. \label{state}
\end{align}
Then, let us quickly verify that the resulting realization indeed achieves the quantum value of $\mathcal{F}_d$. Evaluating $\mathcal{C}_d$ given in Eq.~\eqref{score} leads to
\begin{align}
\mathcal{C}_d =\sum_x \lambda_x \braket{\Phi^+_d|\, A_x\otimes(P_{x_1}-Q_{x_2})\, |\Phi^+_d},
\end{align}
where the summation over $y$ has already been performed to give $ A_x\otimes(P_{x_1}-Q_{x_2}) $. Plugging in Eq.~\eqref{alicesmeas} and exploiting the fact that for any linear operator $L$ we have $L\otimes\mathds{1}\ket{\Phi^+_d}=\mathds{1}\otimes L^\textsf{T}\ket{\Phi^+_d}$, leads to
\begin{align}
\mathcal{C}_d = \sum_x \braket{\Phi^+_d|\, \mathds{1}\otimes(P_{x_1}-Q_{x_2})^2\, |\Phi^+_d}. \label{intermediate step}
\end{align}
Proceeding similarly for $\mathcal{F}_d$, we get
\begin{align}
\mathcal{F}_d=\mathcal{C}_d-\frac{1}{2}\sum_x \lambda_x^2\, \braket{\Phi^+_d|\, (A_x^{(1)}+A_x^{(2)})\otimes\mathds{1}\, |\Phi^+_d}.
\end{align}
Since the measurement operators of Alice are projective, the term in parentheses can be identified as $A_x^{(1)}+A_x^{(2)}=A_x^2$ and
\begin{align}
\mathcal{F}_d &=\mathcal{C}_d-\frac{1}{2}\sum_x \lambda_x^2\, \braket{\Phi^+_d|\, A_x^2\otimes\mathds{1}\, |\Phi^+_d} \nonumber\\
&=\frac{1}{2}\sum_x \braket{\Phi^+_d|\, \mathds{1}\otimes(P_{x_1}-Q_{x_2})^2\, |\Phi^+_d}.
\end{align}
Finally, since $P_{x_1}$ and $Q_{x_2}$ are projectors, we can evaluate the sum:
\begin{align}
\sum_x (P_{x_1}-Q_{x_2})^2 = 2(d-1)\mathds{1}, \label{projectivity}
\end{align}
which leads to $\mathcal{F}_d=d-1$ for the desired realization.

\section{Characterizing the optimal realization}

\noindent Having computed the quantum value of $\mathcal{F}_d$, it would be natural to determine its local value and show that our functionals are indeed non-trivial. However, since $\beta_L$ does not have a closed-form expression, let us postpone it to the next section. Fortunately, it turns out that proving $\beta_Q>\beta_L$ is possible by characterising the realizations that achieve $\beta_Q$. Therefore, to demonstrate that $\mathcal{F}_d$ is non-trivial, we continue our analysis with the device-independent certification of the measurements and the state. 

\begin{theorem}
For any $d\ge 2$ and overlap matrix satisfying $O_{x_1x_2}<1$, under the assumption that Bob's marginal state is full-rank, for any quantum realization which achieves the quantum value of $\mathcal{F}_d$, it holds that:
\begin{enumerate}
\item Bob's measurement operators denoted by $P_{x_1}$ and $Q_{x_2}$ associated with inputs $y=1$ and $2$, respectively, must satisfy
\begin{align}
\begin{aligned}
O_{x_1x_2}^2P_{x_1}&=P_{x_1}Q_{x_2}P_{x_1}\;\text{and}\\ O_{x_1x_2}^2Q_{x_2}&=Q_{x_2}P_{x_1}Q_{x_2}\quad\forall~x_1,~x_2.
\end{aligned} \label{equations}
\end{align}
\item If $O$ has a row (or column) whose entries are all non-zeros, then there exist local isometries $V_A : \mathcal{H}_A \rightarrow \mathbb{C}^d\otimes \mathcal{H}_A$ and $V_B : \mathcal{H}_B \rightarrow \mathbb{C}^d\otimes \mathcal{H}_B$ whose action on the unknown state $\rho_{AB}$ yields
\begin{align}
(V_A \otimes V_B)\rho_{AB}(V_A^\dagger \otimes V_B^\dagger) = \Phi_d^+ \otimes \rho_\mathrm{aux},
\end{align}
where $\Phi_d^+$ is the density matrix of $\ket{\Phi_d^+}$ and $\rho_\mathrm{aux}$ corresponds to the uncharacterized part of $\rho_{AB}$.
\end{enumerate} \label{certification theorem}
\end{theorem}

\begin{proof}[Sketch of the proof] Here we present a sketch just for the first part (see Appendix \ref{proof theorem 1} for a complete proof). To prove an upper bound on the quantum value of $\mathcal{F}_d$ the Cauchy--Schwarz inequality is used several times. Saturating these inequalities allows us to deduce that for all $x = x_{1} x_{2}$ the action of Alice's and Bob's operators on $\rho_{AB}$ satisfies the following relation:
\begin{align}
\lambda_x A_x\otimes \mathds{1}\,\rho_{AB}=\mathds{1}\otimes (P_{x_1}-Q_{x_2})\,\rho_{AB}. \label{proportion}
\end{align}
Since our goal is to certify Bob's measurements, we can assume that Alice's measurements are projective, which implies that $A_x^3=A_x$. Combining this identity with Eq.~\eqref{proportion} yields a polynomial equation in terms of operators $P_{x_1}$ and $Q_{x_2}$ and the marginal state on Bob's side $\rho_{B}$. Assuming that $\rho_{B}$ is full-rank allows us to remove the state dependence. Finally, by examining Eq.~\eqref{projectivity}, we conclude that saturating the quantum value on a state whose $\rho_{B}$ is full-rank requires the measurements of Bob to be projective. This leads to the desired Eqs.~\eqref{equations}.
\end{proof}

Eqs.~\eqref{equations} are derived purely from the fact that we observe the quantum value of $\mathcal{F}_d$. To prove that the functional $\mathcal{F}_d$ is non-trivial, it suffices to argue that these relations cannot be satisfied by any deterministic strategy.

\begin{theorem}
For any $d\ge 2$ and overlap matrix satisfying $O_{x_1x_2}<1$, the quantum value is strictly bigger than the local value. \label{main result}
\end{theorem}

\begin{proof} A deterministic strategy can always be written as a quantum strategy where the local Hilbert spaces are one-dimensional. If Bob outputs $b = u$ for $y=1$ and $b=v$ for $y=2$, the associated projectors correspond to $P_{x_1} = \delta_{x_1u}$ and $Q_{x_2} = \delta_{x_2v}$. If this strategy saturates the quantum value,~Eqs.~\eqref{equations} would imply that $O_{x_1x_2}^2\delta_{x_1u}=\delta_{x_1u}\delta_{x_2v}\delta_{x_1u}$, for all $x_{1}, x_{2}$. However, choosing $x_1=u$ and $x_2=v$ leads to $O_{uv} = 1$, which is a contradiction.
\end{proof}

At this point, we can see why the assumption $O_{x_1x_2}<1$ was necessary. It is straightforward to check that $\mathcal{F}_d$ can be defined for any overlap matrix, the quantum value is always equal to $d - 1$ and the conclusion given in Eqs.~\eqref{equations} holds. However, without the assumption that $O_{x_1x_2}<1$, the measurement certification statement is no longer sufficient to deduce that the functional is non-trivial. In fact, in Appendix \ref{proof of the lower bound theorem}, we show that whenever $O_{x_{1} x_{2}} = 1$ for some $x$, then there exists a deterministic strategy achieving $\beta_{Q}$, i.e.~$\beta_{L} = \beta_{Q}$.

As explained at the beginning, every pair of incompatible rank-one projective measurements can be preprocessed to give measurements whose overlap matrix satisfies $O_{x_{1} x_{2}} < 1$ and, hence, gives rise to a non-trivial functional. Therefore, we obtain the following corollary.

\begin{corollary}
Every pair of incompatible rank-one projective measurements on a finite-dimensional Hilbert space is capable of producing non-local correlations and is optimal for some non-trivial Bell inequality. \label{main corollary}
\end{corollary}

\section{Robustness against noise}

\noindent To investigate how robust to noise the realizations presented above are, consider a specific noise model in which the measurements are kept unchanged while the state is replaced by the isotropic state
\begin{align}
\rho_\nu=\nu\Phi^+_d+(1-\nu)\frac{\mathds{1}\otimes\mathds{1}}{d^2}, \label{isotropic}
\end{align}
where $\nu\in[0,\, 1]$ is the visibility. A simple calculation shows that performing the optimal measurements on the isotropic state gives the value of $(d-1)(2\nu-1)$. We can use this result to compare the noise robustness of different functionals $\mathcal{F}_d$. More specifically, we are interested in discovering which functional enables a violation for the smallest $\nu$. To address this question, we first need to compute the local value.

For a fixed overlap matrix $O$, calculating the local value reduces to analysing the deterministic strategies for $\mathcal{F}_d$. If we choose the strategy of Bob in which he outputs  $b=u$ if $y=1$ and $b=v$ if $y=2$, the value becomes
\begin{align}
\mathcal{F}_d=\sum_x \lambda_x\left(\delta_{x_1u}-\delta_{x_2v}\right)A_x-\frac{1}{2}\sum_x\lambda_x^2A_x^2, \label{local bound 1}
\end{align}
where $A_x\in\{\pm 1,\, 0\}$ describes the deterministic strategy of Alice. Let us now show that the optimal strategy of Alice can be explicitly determined. The sum can be split into three distinct cases: $R_{\pm}=\{x\in [d]^2\, |\, (\delta_{x_1u}-\delta_{x_2v})=\pm 1\}$ and $R_0=[d]^2 \,\backslash\, (R_+\cup R_-)$. For the $R_0$ terms, the optimal strategy of Alice is clearly not to play the game (this would lead to a negative contribution), so the optimal choice is $A_x=0$.

For the $R_{\pm}$ terms the optimal choice is $A_{x} = \delta_{x_{1} u} - \delta_{x_{2} v}$, which is better than $A_{x} = - ( \delta_{x_{1} u} - \delta_{x_{2} v} )$ (because $\lambda_{x} > 0$) and better than $A_{x} = 0$ (because $\lambda_x-\frac{1}{2}\lambda_x^2 > 0$). Plugging in the optimal strategy of Alice into Eq.~\eqref{local bound 1} leads to the following auxiliary function:
\begin{align}
s(u,\, v) : = \sum_{x\in R_{\pm}} \left(\lambda_x-\frac{1}{2}\lambda_x^2\right), \label{local bound 2}
\end{align}
where the dependence on $u$ and $v$ is hidden inside the definitions of $R_{\pm}$. Then, $\beta_L $ can be written as
\begin{align}
\beta_L(O) =\max_{u,\, v} \: \left[s(u, v)\right]. \label{local value}
\end{align}
The choice of coefficients of our functionals, made at the beginning, ensures that the quantum value does not depend on the overlap matrix. It is a convenient choice because all the dependence of $O$ is contained in the local value. It is easy to see that in the noise model presented above, the largest robustness corresponds to the lowest local value. Searching for highly-robust functionals leads to the following theorem.

\begin{theorem}
For $d\ge 2$ and any overlap matrix, $\beta_L(O)\ge d+\sqrt{2}-\sfrac{5}{2}$. \label{lower bound theorem}
\end{theorem}

\begin{proof}[Sketch of the proof](See Appendix \ref{proof of the lower bound theorem} for details). Primarily, note that a lower bound for any strategy $s(u,~v)$ is also a lower bound for $\beta_L(O)$. Because we can always relabel the measurements outputs, we assume, for simplicity, that $O_{11}$ is the largest element of $O$, i.e. $O_{11}\ge O_{uv}$, $\forall~u,~v \in \{1,\, ...,\, d\}$. Then, we lower bound $\beta_L(O)$ using the first strategy:
\begin{align}
\beta_L(O)= \max_{u,\, v}\: [s(u,\, v)]\ge s(1,\, 1).
\end{align}
Next, as $s(1,\, 1)$ depends only on overlaps squared, we define a new variable $t_{x_1x_2}:=O_{x_1x_2}^2$. This makes some terms in $s(1,\, 1)$ strictly concave in $t_{x_1x_2}$, which allows us to obtain a lower bound that depends only on $t_{11}$. Minimising that function concludes the proof.
\end{proof}

Although Theorem \ref{lower bound theorem} provides a lower bound on $\beta_L(O)$, it does not say whether this bound is achievable. Fortunately, by demanding the saturation of the strictly concave terms of $s(1,\, 1)$, we extract this information.

\begin{lemma} \label{optimal matrix}
For even $d \ge 2$, the lower bound of Theorem \ref{lower bound theorem} is achievable and every $O$ that saturates it corresponds to a direct sum of qubit MUBs, i.e., up to permutations, it can be written into as 
\begin{align}
O=\bigoplus_{i=1}^{\sfrac{d}{2}}\frac{1}{\sqrt{2}}J_2, \label{even matrix}
\end{align}
where $J_2$ is the $2\times 2$ matrix of ones. For odd $d \ge 3$, the lower bound of Theorem \ref{lower bound theorem} cannot be achieved by any overlap matrix.
\end{lemma}

The complete proof of Lemma \ref{optimal matrix} can be found in Appendix \ref{proof lemma 2}. Let us briefly comment on the difference between even and odd $d$. Our proof of Theorem \ref{lower bound theorem} is based solely on the fact that the rows and columns of $O$ are normalized. However, from Eq.~\eqref{overlap}, we know that $O$ should be obtained by taking the entry-wise absolute value of a unitary matrix, which, in general, is more restrictive than just imposing the normalisation condition. For even $d$, there exists a valid overlap matrix which saturates the lower bound obtained in Theorem \ref{lower bound theorem}, but for odd $d$ this is not the case. For a more detailed explanation, see Appendix \ref{explanation}. 

\section{Discussion}

\noindent In this work, we have demonstrated how to tailor a Bell functional to a specific pair of incompatible rank-one projective measurements. More specifically, we have shown that for every such pair there exists a non-trivial Bell functional for which this pair is optimal. We have also proved that for these functionals a certain degree of certification is possible. If the quantum value of $\mathcal{F}_d$ is achieved, the measurements of Bob must satisfy some simple polynomial relations and, under some additional mild conditions, a maximally entangled state can be extracted.

We have also investigated how robust to noise these functionals are, as quantified by the difference $\beta_{Q} - \beta_{L}$. For even $d$, the largest robustness to noise arises when the measurements are chosen as direct sums of qubit MUBs. For odd $d$, we have a conjecture for $d = 3$, but no analytical proof of optimality. It is worth pointing out, however, that for $d \ge 3$ the optimal noise robustness is not achieved by $d$-dimensional MUBs.

There are several open questions arising from this work. An aspect that should be clarified are the precise conditions under which state certification is possible. We have shown that if the overlap matrix contains a non-null row or column, then a maximally entangled state can be extracted. On the other extreme, there are overlap matrices that decompose as a direct sum for which, as proved in Appendix \ref{exceptional cases}, state certification is not possible. Hence, the question about state certification is open for intermediate cases. Another question concerns the optimal noise robustness for odd dimensions. Our candidate for $d = 3$ exhibits a mathematically elegant structure and it would be interesting to see if the same happens for larger dimensions.

Lastly, it is known that all pure entangled states can generate non-locality, and we have demonstrated that incompatible pairs of rank-one projective measurements also do so. By assembling these results, is it possible to achieve necessary and sufficient conditions for, at least, some limited class of states \emph{and} measurements?

\section{Acknowledgements}

We acknowledge fruitful discussions with M{\'a}t{\'e} Farkas and Nicolas Gigena. The project ``Robust certification of quantum devices'' is carried out within the HOMING programme of the Foundation for Polish Science co-financed by the European Union under the European Regional Development Fund.
\bibliographystyle{apsrev4-1}
\bibliography{references.bib}

\onecolumngrid

\appendix

\titleformat{\section}{\small\bfseries}{Appendix \thesection:}{0.5em}{\centering}
\titleformat{\subsection}{\small\bfseries}{\thesection.\,\thesubsection}{1em}{\centering} 

\section{The quantum value of \texorpdfstring{$\mathcal{F}_d$}{Lg}} 

\noindent In this appendix, we provide a complete proof of Lemma \ref{lemma1}. In the main text, we specified a quantum realization for the functional $\mathcal{F}_d$ which achieves the value $d-1$. Here, in the first part of this appendix, we provide a matching upper bound, which shows that the quantum value of $\mathcal{F}_d$ equals $d-1$. In the second part, we present a natural extension of our functional to a set of $N>2$ rank-one projective measurements.

\subsection{Proof of Lemma \ref{lemma1}} \label{proof lemma 1}

\noindent At this point, we are only interested in deriving an upper bound, hence, we could, without loss of generality, restrict ourselves to realizations where the state is pure and measurements are projective. However, we will later reuse this argument to derive some certification statements, so we do not want to introduce unnecessary assumptions. As we are not interested in characterizing the measurements of Alice, we will assume that they are projective, but we make no assumptions about the measurements of Bob and the shared state. Thus, let us represent the state by $\rho_{AB}$, an arbitrary density matrix, and reuse the symbols $A_x$, $P_{x_1}$ and $Q_{x_2}$ to denote the observables of Alice and the two measurements of Bob, respectively.

We start again by using the Born rule to rewrite $\mathcal{C}_d$ as
\begin{align}
\mathcal{C}_d=\sum_x \lambda_x \,\tr \left[\left(A_x\otimes(P_{x_1}-Q_{x_2})\right)\rho_{AB}\right]. \label{C score}
\end{align}
Taking the absolute value of every term in the summation leads to
\begin{align}
\mathcal{C}_d \le \sum_x \left|\lambda_x \,\tr \left[\left(A_x\otimes(P_{x_1}-Q_{x_2})\right)\rho_{AB}\right]\right|. \label{summation}
\end{align}
In the next step, we apply the Cauchy--Schwarz inequality. For the Hilbert--Schmidt inner product, the Cauchy--Schwarz inequality reads $\left|\tr(X^\dagger Y)\right|\le \sqrt{\tr(X^\dagger X)}\,\sqrt{\tr(Y^\dagger Y)}$, where $X$ and $Y$ are arbitrary operators acting in a given Hilbert space. Therefore, taking $X=\left(A_x\otimes\mathds{1}\right)\rho_{AB}^{\sfrac{1}{2}}$ and $Y=\left(\mathds{1}\otimes\left(P_{x_1}-Q_{x_2}\right)\right)\rho_{AB}^{\sfrac{1}{2}}$, where $\rho_{AB}^{\sfrac{1}{2}}$ is the positive semi-definite square root of $\rho_{AB}$, we can apply the Cauchy--Schwarz inequality to all of the terms of $\mathcal{C}_d$, leading to
\begin{align}
\mathcal{C}_d &\le \sum_x \lambda_x \sqrt{\tr\left[\left(A_x^2\otimes\mathds{1}\right)\rho_{AB}\right]}\sqrt{\tr\left[\left(\mathds{1}\otimes(P_{x_1}-Q_{x_2})^2\right)\rho_{AB}\right]}. \label{score sat}
\end{align}
The second application of Cauchy--Schwarz, now for the standard inner product of real vectors, leads to
\begin{align}
\mathcal{C}_d\le \sqrt{\sum_x \lambda_x^2\,\tr\left[\left(A_x^2\otimes\mathds{1}\right)\rho_{AB}\right]}\sqrt{\sum_x\tr\left[\left(\mathds{1}\otimes(P_{x_1}-Q_{x_2})^2\right)\rho_{AB}\right]}. \label{score sat 2}
\end{align}
Let us now compute a universal upper bound on the second factor. For any measurement operators $P_{x_1}$ and $Q_{x_2}$, it is true that $P_{x_1}^2\le P_{x_1}$ and $Q_{x_2}^2\le Q_{x_2}$ (recall that for Hermitian operators $X, Y$ the inequality $X \geq Y$ is equivalent to $X - Y \geq 0$). Thus,
\begin{align}
\sum_{x_1,x_2}(P_{x_1}-Q_{x_2})^2 &= \sum_{x_1,x_2}(P_{x_1}^2+Q_{x_2}^2-\{P_{x_1},\, Q_{x_2}\}) \nonumber \\
&\le \sum_{x_1,x_2}(P_{x_1}+Q_{x_2}-\{P_{x_1},\, Q_{x_2}\}) = 2(d-1)\mathds{1}. \label{projec sat}
\end{align}
Therefore, the second factor is upper bounded by $\sqrt{2(d-1)}$. Using the assumption that the measurement operators of Alice are projective, we can rewrite the observable as $A_x^2= A_x^{(1)}+A_x^{(2)}$. Therefore,
\begin{align}
\mathcal{C}_d\le \sqrt{2(d-1)\sum_x \lambda_x^2\,\tr\left[\left(\left(A_x^{(1)}+A_x^{(2)}\right)\otimes\mathds{1}\right)\rho_{AB}\right]}. 
\end{align}
Finally, if we define
\begin{align}
\gamma:=\sum_x \lambda_x^2\,\tr\left[\left((A_x^{(1)}+A_x^{(2)})\otimes\mathds{1}\right)\rho_{AB}\right], \label{gamma}
\end{align}
it is easy to see that $\mathcal{F}_d$ is upper bounded by
\begin{align}
\mathcal{F}_d\le \sqrt{2(d-1)\gamma}-\frac{1}{2}\, \gamma. \label{quantum bound}
\end{align}
The last step is to maximize the right-hand side of Eq.~\eqref{quantum bound} over $\gamma = [0, d(d-1)]$  (it is easy to verify that the maximal value of $\gamma$ equals $\sum_{x} \lambda_{x}^{2} = d(d-1)$). Using elementary calculus we conclude that the right-hand side is maximized only when $\gamma=2(d-1)$, which leads to the final bound:
\begin{align}
    \mathcal{F}_d\le d-1. \label{quantum value}
\end{align}

At this point, the reader can better appreciate why we needed to introduce a penalty in Eq.~\eqref{functional}. In Eq.~\eqref{quantum bound}, $\gamma$ is maximized over the interval $[0, d(d - 1)]$. Then, the maximal $\gamma$ can be found at $2(d - 1)$, giving us the quantum value above. Now, note that if our functional had no penalty, it would be written as $\mathcal{F}_d = \mathcal{C}_d$ and the maximization of $\mathcal{C}_d$ in terms of $\gamma$ would produce
\begin{align}
	\mathcal{C}_d \le \sqrt{2(d - 1)\gamma}. \label{functional without penalty}
\end{align} 
In this case, $\gamma$ is maximal at $d(d-1)$, in the edge of the interval, so $\mathcal{C}_d$ is upper bounded by $\sqrt{2d}(d - 1)$.  On the other hand, by comparing Eqs.~\eqref{intermediate step} and \eqref{projectivity}, we obtain that the realization produced by the maximally entangled state is simply $\mathcal{C}_d = 2(d - 1)$. In a nutshell, introducing the penalty over the outputs $a = 1$ and $2$, force the realization by the maximally entangled state the same as the upper bound for $\mathcal{F}_d$ so that this state turns out to be optimal. 

In addition, note that, without penalty, Alice never outputs $a = \perp$. Then, the inclusion of a third output is justified. It allows Alice to play the game and be penalized or not to play the game and keep the value of the realization unchanged. Adding that a third output also guarantees that the spectrum of $A_x$ is proportional to the spectrum of $P_{x_1} - Q_{x_2}$, since $\text{spec}\, (P_{x_1} - Q_{x_2}) \in \{\pm \lambda_x,~0\}$.

\subsection{An extension to \texorpdfstring{$N > 2$}{Lg} rank-one projective measurements}

\noindent The original functionals tailored to MUBs in Ref.~\cite{Tavakoli21} have recently been extended to $N > 2$ bases \cite{Colomer22} and it turns out that an analogous extension can be constructed for non-MUB measurements. In this subsection, we present this generalization.

Suppose a bipartite Bell scenario in which Alice is given a string $x = (j,\, k)\, x_1 x_2 \, ... \, x_N$, in which $x_i \in [d]$, for $i = 1,~...,~N$, and $(j,\, k)$ is an ordered pair such that $j \in [N] - \{N\}$, $k \in [N] - \{1\}$ and $j < k$. Similarly to the previous case, Alice now outputs $a \in \{ 1,\, 2,\, ...,\, N,\, \perp \}$. On the other hand, the outputs of Bob are kept unchanged, while the possible inputs are $y \in [N]$.
\begin{figure}[ht!]
\includegraphics[scale = 1]{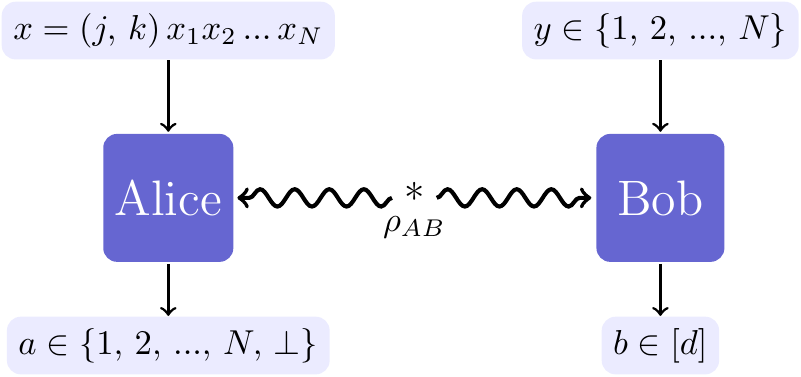}
\caption{A Bell scenario extended to a set of $N > 2$ projective measurements. Now, Alice is given a string with $N$ digits $x_1,~ x_2,~...,~ x_N$ ranging from $1$ to $d$ and an ordered pair $(j,\, k)$, such that  $j \in [N] - \{N\}$ and $k \in [N] - \{1\}$. Similarly, Alice's now output ranges from $1$ to $N$ and includes $\perp$. Bob takes an input $y \in  [N]$ and produces output $b \in [d]$.} \label{extended scenario}
\end{figure}

Now, consider $N$ $d$-dimensional orthonormal bases $\{\ket{e_{x_1}^{(1)}}\}_{x_1 = 1}^d$, $\{\ket{e_{x_2}^{(2)}}\}_{x_2 = 1}^d$, ..., $\{\ket{e_{x_N}^{(N)}}\}_{x_N = 1}^d$, where $d \ge 2$, and define
\begin{align}
\lambda_x^{(j,k)} := \sqrt{1 - |\braket{e_{x_j}^{(j)}| e_{x_k}^{(k)}}|^2}. \label{new lambda}
\end{align}
For a fixed pair $(j,~k)$, consider the following functional
\begin{align}
	\mathcal{F}_d^{(j,k)} := \sum_{x_j,x_k = 1}^d \, \lambda^{(j,k)}_x \Bigg(
	\sum_{y \in \{j, k\}} \left[ \, p(y,\, x_y | x,\, y) - p(\bar{y},\, x_y | x,\, y) \right]
	- \frac{1}{2} \lambda^{(j,k)}_x \, \left[ \, p(a = j|x)+p(a = k|x) \right] \Bigg),
\end{align}
where $\bar{y}$ flips the value of $y$ from $j$ to $k$ and vice-versa. Then, $\mathcal{F}_d$ can be generalized by considering
\begin{align}
	\mathcal{F}_d^N = \sum_{j < k}^N \mathcal{F}_d^{(j,k)}. \label{new functional}
\end{align}
If written so, Eq.~\eqref{new functional} is clearly upper bounded by
\begin{align}
	\mathcal{F}_d^N \le \frac{1}{2}N(N - 1)(d - 1), \label{new quantum value}
\end{align}
as we can simply use Eq.~\eqref{quantum value} to sum the $\binom{N}{2}$ terms of Eq.~\eqref{new functional}.

This upper bound can be saturated if the $N$ orthonormal bases $\{\ket{e_{x_i}^{(i)}}\}_{x_i = 1}^d$ in dimension $d$ are such that $\lambda^{(j,k)}_x > 0$, for all $j,~k$, $x$ and $i = 1,\, 2,\, ...,\, N$. In this case, we can mirror the optimal realization proposed in the main text by taking the same state as in Eq.~\eqref{state}, $P_{x_i} = \ket{e_{x_i}^{(i)}} \!\! \bra{e_{x_i}^{(i)}}$ when Bob is given $y = i$ and
\begin{align}
A_x^{(j,k)} = \frac{1}{\lambda_x^{(j,k)}} (P_{x_j} - Q_{x_k})^\mathsf{T}.
\end{align}
Moreover, according to Theorem \ref{main result} the functional in Eq.~\eqref{new functional} is non-trivial, if, for the $N$ orthonormal bases, the condition $\lambda^{(j,k)}_x > 0$, $\forall~j,~k,~x$, is satisfied.

Lastly, an interesting problem can be posed if, instead of providing $N$ orthonormal bases, we specify $\lambda^{(j,k)}_x$ and try to recover the bases. In this case, it is not clear if this problem has a solution for any $N$ or arrangements of  $\lambda^{(j,k)}_x$. For instance, if we consider
\begin{align}
	|\braket{e_{x_j}^{(j)} | e_{x_k}^{(k)}}|^2 = \frac{1}{d},\quad \forall \, j,~k~\in [N]~\text{and} ~x_j, x_k \in [d],
\end{align}
we fix the values of $\lambda^{(j,k)}_x$ to be uniform. The resulting functionals are the same as the ones derived in Ref.~\cite{Colomer22} up to a factor of $\sqrt{1 - \sfrac{1}{d}}$. The numerical results of Ref.~\cite{Colomer22} confirm what is already known in the literature: there is no more than $d + 1$ MUBs in dimension $d$ and the Zauner's conjecture \cite{Zauner99} holds.

\section{Certification} \label{proof theorem 1}

\noindent In what follows, we present the proof of Theorem \ref{certification theorem}. First, we demonstrate the device-independent certification of the measurements of Bob, based solely on the condition that $\beta_Q$ is saturated. Recall that, as in Appendix \ref{proof lemma 1}, for the certification of the measurements of Bob, we can assume that the measurements of Alice are projective since we are not interested in their characterization. Next, we construct an isometry capable of extracting a maximally entangled state of dimension $d$. Lastly, we discuss some exceptional cases in which state certification cannot be achieved.

\subsection{Bob's measurements (proof of the first part of Theorem \ref{certification theorem})}

\noindent  If Ineq.~\eqref{quantum value} is saturated, all of the inequalities used in Appendix \ref{proof lemma 1} should be tight. This enables us to get some information about the state and measurements capable of achieving the quantum value.

The first step we must consider is going from Eq.~\eqref{C score} to Ineq.~\eqref{summation}, where the terms of $\mathcal{C}_d$ are bounded by the absolute value:
\begin{align}
\mathcal{C}_d = \sum_x \lambda_x\, \tr\left[(A_x\otimes(P_{x_1}-Q_{x_2}))\rho_{AB}\right] \le \sum_x \left|\lambda_x\, \tr\left[(A_x\otimes(P_{x_1}-Q_{x_2}))\rho_{AB}\right]\right|.
\end{align}
If the above inequality is tight, the argument of the absolute value must be non-negative. Since, by construction, $\lambda_x> 0$, we conclude that $\tr\left[A_x\otimes(P_{x_1}-Q_{x_2})\rho_{AB}\right]\ge 0$.

The saturation condition of the Cauchy--Schwarz inequality in Ineq.~\eqref{score sat}, in turn, implies that,
\begin{align}
\lambda_x\, A_x\otimes \mathds{1}\,\rho_{AB}^{\sfrac{1}{2}} = \mu_x\mathds{1}\otimes (P_{x_1}-Q_{x_2})\,\rho_{AB}^{\sfrac{1}{2}}, \quad \forall\: x= x_1 x_2, \label{saturation 0}
\end{align}
where $\mu_x\in\mathbb{C}$ is the proportionality constant. Right-multiplying it by $\rho_{AB}^{\sfrac{1}{2}}$ and left-multiplying by $\mathds{1}\otimes(P_{x_1}-Q_{x_2})$ leads to
\begin{align}
    \lambda_x\, A_x\otimes (P_{x_1}-Q_{x_2})\,\rho_{AB} = \mu_x\mathds{1}\otimes (P_{x_1}-Q_{x_2})^2\,\rho_{AB}.
\end{align}
If we trace both sides of this equation and use the fact that $\tr\left[A_x\otimes(P_{x_1}-Q_{x_2})\rho_{AB}\right]\ge 0$ (deduced above), we conclude that $\mu_x$ is real and non-negative.

The saturation condition of Ineq.~\eqref{score sat 2} implies that
\begin{align}
\lambda_x^2\, A_x^2\otimes\mathds{1}\,\rho_{AB}=\eta\,\mathds{1}\otimes (P_{x_1}-Q_{x_2})^2\,\rho_{AB}\, \quad \forall\: x=x_1,\, x_2, \label{saturation 1}
\end{align}
where $\eta$ is independent of $x$ since the Cauchy--Schwarz in Ineq.~\eqref{score sat 2} is applied to the entire summation. Now, let us multiply Eq.~\eqref{saturation 0} by its Hermitian conjugate and take the trace:
\begin{align}
    \lambda_x^2\, \tr \left[A_x^2\otimes \mathds{1}\,\rho_{AB}\right] = \mu_x^2\,\tr\left[\mathds{1}\otimes (P_{x_1}-Q_{x_2})^2\,\rho_{AB}\right] \label{saturation}
\end{align}
Comparing Eqs.~\eqref{saturation} and \eqref{saturation 1} we conclude that $\mu_{x}=\sqrt{\eta}$, for all $x$. So, let us drop the index of $\mu_x$ and use just $\mu$ from now on.

The last but not less important saturation condition is implicit in Ineq.~\eqref{projec sat} where we upper bound the following quantity by:
\begin{align}
\sum_x \tr \left[\left(\mathds{1} \otimes (P_{x_1}^2+Q_{x_2}^2-\{P_{x_1},\, Q_{x_2}\})\right)\rho_{AB}\right] \le \sum_x \tr \left[\left(\mathds{1}\otimes (P_{x_1}+Q_{x_2}-\{P_{x_1},\, Q_{x_2}\})\right) \rho_{AB}\right].
\end{align}
After a short algebraic manipulation, we get
\begin{align}
\tr\left[\left(\mathds{1} \otimes \sum_{x_1} (P_{x_1}-P_{x_1}^2)\right)\rho_{AB}\right]+\tr\left[\left(\mathds{1} \otimes \sum_{x_2} (Q_{x_2}-Q_{x_2}^2)\right)\rho_{AB}\right]\ge 0.
\end{align}
If we trace over Alice's subsystem, the term-by-term saturation of the above inequality implies
\begin{align}
\begin{aligned}
\tr\left[(P_{x_1}-P_{x_1}^2)\rho_B\right]&= 0,\quad\forall\, x_1 \\
\tr\left[(Q_{x_2}-Q_{x_2}^2)\rho_B\right]&= 0,\quad\forall\, x_2.
\end{aligned}
\end{align}
Since both $P_{x_1}-P_{x_1}^2$ (or $Q_{x_2}-Q_{x_2}^2$) and $\rho_B$ are positive semi-definite operators and their product has a null trace, they must be orthogonal:
\begin{align}
\begin{aligned}
(P_{x_1}-P_{x_1}^2)\rho_B &= 0, \\
(Q_{x_2}-Q_{x_2}^2)\rho_B &= 0.
\end{aligned} \label{projection equations}
\end{align}
Assuming that Bob's marginal state is full rank, it is possible to eliminate $\rho_B$ out of Eqs.~\eqref{projection equations} by right-multiplying it by $\rho_B^{-1}$. This way, we conclude that the saturation of Ineq.~\eqref{projec sat} implies that both measurements of Bob are projective. We can use this fact to sum Eq.~\eqref{saturation} over $x$ to obtain
\begin{align}
\sum_x \lambda_x^2\, \tr \left[A_x^2\otimes \mathds{1}\,\rho_{AB}\right] &= \mu^2 \sum_x \tr\left[\mathds{1}\otimes (P_{x_1}-Q_{x_2})^2\right]
\end{align}
Note that, if we assume the measurements of Alice to be projective, the right-hand side of the above equation can be identified as $\gamma$, in Eq.~\eqref{gamma}. On the other hand, the left-hand side can be summed to $2\mu^2(d-1)$. The only value of $\gamma$ which allows for the quantum value in Eq.~\eqref{quantum bound} is $2(d-1)$, so we conclude that $\mu=1$.

Having obtained the value of the proportionality constants, let us right-multiply Eq.~\eqref{saturation 0} by $\rho_{AB}^{\sfrac{1}{2}}$ and use $\mu=1$ to rewrite the following relation:
\begin{align}
\lambda_x A_x\otimes \mathds{1}\,\rho_{AB}=\,\mathds{1}\otimes (P_{x_1}-Q_{x_2})\,\rho_{AB}. \label{Alice and Bob measurements}
\end{align}
Again, assuming that the measurements of Alice are projective, we can use the Eq.~\eqref{alicesobs} to write $A_x^3=A_x$. Combining it with Eq.~\eqref{Alice and Bob measurements}, we get
\begin{align}
\frac{1}{\lambda_x}\,\mathds{1}\otimes (P_{x_1}-Q_{x_2})\,\rho_{AB}=\frac{1}{\lambda_x^3}\,\mathds{1}\otimes (P_{x_1}-Q_{x_2})^3\,\rho_{AB}.
\end{align}
Since we have eliminated Alice's observable out of the equation, we can derive a relation involving only Bob's measurement operators by tracing out Alice's subsystem and right-multiplying by the inverse of Bob's marginal state:
\begin{align}
\lambda_x^2 (P_{x_1}-Q_{x_2})=(P_{x_1}-Q_{x_2})^3. \label{similar}
\end{align}
After some algebraic manipulation,
\begin{align}
O_{x_1x_2}^2(P_{x_1}-Q_{x_2})=P_{x_1}Q_{x_2}P_{x_1}-Q_{x_2}P_{x_1}Q_{x_2}.
\end{align}
Now, if we sum over $x_1$,
\begin{align}
\sum_{x_1} O_{x_1x_2}^2P_{x_1}-\sum_{x_1}O_{x_1x_2}^2Q_{x_2} &=\sum_{x_1} (P_{x_1}Q_{x_2}P_{x_1}-Q_{x_2}P_{x_1}Q_{x_2}) \\
\sum_{x_1} O_{x_1x_2}^2P_{x_1}-Q_{x_2} &=\sum_{x_1} (P_{x_1}Q_{x_2}P_{x_1})-Q_{x_2}^2.
\end{align}
Because $Q_{x_2}=Q_{x_2}^2$,
\begin{align}
\sum_{x_1} O_{x_1x_2}^2P_{x_1} &=\sum_{x_1} (P_{x_1}Q_{x_2}P_{x_1})
\end{align}
Finally, since the $P_{x_1}$ projectors are orthogonal, this equality must hold term by term:
\begin{align}
O_{x_1x_2}^2P_{x_1}=P_{x_1}Q_{x_2}P_{x_1}. \label{equation1}
\end{align}
The summation over $x_2$ gives us the complementary equation:
\begin{align}
O_{x_1x_2}^2Q_{x_2}=Q_{x_2}P_{x_1}Q_{x_2}. \label{equation2}
\end{align}

Lastly, let us show that these two equations are complete i.e., if for some finite-dimensional measurements Eqs.~\eqref{equation1} and \eqref{equation2} hold, then it is possible to construct a realization that achieves the quantum value of $\mathcal{F}_d$. This means that the conditions derived are, at least for finite-dimensional measurements, necessary and sufficient, and no tighter characterisation can be obtained.

Consider an incompatible pair of measurements $\{P_{x_1}\}_{x_1=1}^d$ and $\{Q_{x_2}\}_{x_2=1}^d$ acting on a $D$-dimensional Hilbert space, with $D<\infty$, satisfying Eqs.~\eqref{equation1} and \eqref{equation2},
for some set of overlaps $0<O_{x_1x_2}<1$. Then, both $P_{x_1}$ and $Q_{x_2}$ are projectors and have equal traces, for all $x_1,~x_2\in\{1,~...,~d\}$. The projectivity can be demonstrated by simply summing Eqs.~\eqref{equation1} and \eqref{equation2} over $x_1$ and $x_2$ to get $P_{x_1}=P_{x_1}^2$ and $Q_{x_2}=Q_{x_2}^2$, respectively. Showing that $\tr(P_{x_1})=\tr(Q_{x_2})$ relies on demonstrating that the projectors are isomorphic. Defining $W_{x_1x_2}=\frac{1}{O_{x_1x_2}}P_{x_1}Q_{x_2}$ we can see that, for $O_{x_1x_2}>0$,
\begin{align}
W_{x_1x_2}W_{x_1x_2}^\dagger = P_{x_1} \quad \text{and} \quad W_{x_1x_2}^\dagger W_{x_1x_2} = Q_{x_2}.
\end{align}
For a finite-dimensional Hilbert space this implies that $\tr(P_{x_1})=\tr(Q_{x_2})=n$, for some fixed $n\in\mathbb{N}$. Because the dimension of the Hilbert space is finite, we rewrite $D$ as simply
\begin{align}
D=\tr\mathds{1}=\sum_{x_1}^d\tr P_{x_1}=dn.
\end{align}

Then, we construct a Hermitian operator $B_{x_1x_2}=P_{x_1}-Q_{x_2}$. By computing an explicit expression for $B_{x_1x_2}^{3}$, one can show that it satisfies a relation analogous to Eq.~\eqref{similar}:
\begin{align}
B_{x_1x_2}^3=\lambda_{x}^2B_{x_1x_2},\quad\forall~ x_1,~x_2\in\{1,~...,~d\}.
\end{align}
Clearly, the spectrum of $B_{x_1x_2}$ belongs to $\{0,\pm\lambda_{x}\}$. Because the trace of $B_{x_1x_2}$ is null, the multiplicity of its non-null eigenvalues must be the same. So, let us calculate
\begin{align}
\tr(B_{x_1x_2}^2)&= \tr\left(P_{x_1}+Q_{x_2}-\{P_{x_1},\, Q_{x_2}\}\right) \nonumber \\
&= 2n - \tr\left(\{P_{x_1},\, Q_{x_2}\}\right) \nonumber \\
&= 2n - 2\tr(P_{x_1}Q_{x_2}P_{x_1})=2n\lambda_{x}^2,
\end{align}
where we have used that $\tr(P_{x_1}Q_{x_2})=\tr(P_{x_1}Q_{x_2}P_{x_1})$. Thus, both $+\lambda_{x}$ and $-\lambda_{x}$ must have multiplicity $n$.

Now, let us define a realization of $\mathcal{F}_d$ in which the measurements of Bob are represented by $P_{x_1}$ and $Q_{x_2}$, the observables of Alice are given by $A_x=\sfrac{1}{\lambda_x}B_{x_1x_2}^\textsf{T}$ and the state is the $D$-dimensional maximally entangled state $\ket{\Phi_D^+}$. Like in Eq.~\eqref{alicesmeas}, the spectrum of $A_x$ also belongs to $\{0,\, \pm 1\}$. Then, evaluating the correlation score gives
\begin{align}
\mathcal{C}_d &= \sum_{x}\lambda_x\bra{\Phi_D^+}A_x\otimes B_{x_1x_2}\ket{\Phi_D^+} = \sum_{x}\bra{\Phi_D^+}\mathds{1}\otimes B_{x_1x_2}^2\ket{\Phi_D^+} \nonumber \\
&= \frac{1}{D}\sum_x\tr(B_{x_1x_2}^2) = \frac{1}{dn} \sum_x 2n\left(1-O_{x_1x_2}^2\right) = 2(d-1),
\end{align}
where we used the fact that the local state of Bob is $\sfrac{\mathds{1}}{D}$. For the final score, we obtain
\begin{align}
\mathcal{F}_d &= \mathcal{C}_d-\frac{1}{2}\sum_x \lambda_x^2 \bra{\Phi_D^+}A_x^2\otimes \mathds{1}\ket{\Phi_D^+} = \mathcal{C}_d-\frac{1}{2}\sum_x \bra{\Phi_D^+}\mathds{1}\otimes B_{x_1x_2}^2\ket{\Phi_D^+} = \mathcal{C}_d-\frac{1}{2D}\sum_x\tr(B_{x_1x_2}^2) = d-1,
\end{align}
as desired.

\subsection{The state (proof of the second part of Theorem 1)}

\noindent Here, we present the argument for the certification of the state. We show that, if the overlap matrix $O$ has at least one column or row whose entries are non-zero, then for any realization that saturates the quantum value of $\mathcal{F}_d$ we can construct local isometries $V_A:\:\mathcal{H}_A\rightarrow\mathbb{C}^d\otimes\mathcal{H}_A$ and $V_B:\:\mathcal{H}_B\rightarrow\mathbb{C}^d\otimes\mathcal{H}_B$ such that:
\begin{align}
(V_A\otimes V_B)\rho_{AB}(V_A^\dagger\otimes V_B^\dagger)=\Phi^+_d\otimes\rho_\text{aux}, \label{extraction}
\end{align}
where $\rho_{AB}\in\mathcal{L}\left(\mathcal{H}_A\otimes\mathcal{H}_B\right)$ is the shared state, $\Phi^+_d$ is the $d$-dimensional maximally entangled state, and $\rho_\text{aux}$ corresponds to the uncharacterized part of $\rho_{AB}$. In particular, note that the isometries depend on the measurement operators of Bob, as explained below.

Let us start by defining the isometry on Bob's side. To do this, we need to assume that the $j$-th column of $O$ is non-zero, i.e. $O_{i,j}\neq 0$, for all $i$. Then, we introduce the isometry $V_B: = SR$, where $R:\:\mathcal{H}_B\rightarrow\mathbb{C}^d\otimes\mathcal{H}_B$ and $S:\:\mathbb{C}^d\otimes\mathcal{H}_B\rightarrow\mathbb{C}^d\otimes\mathcal{H}_B$ are defined by
\begin{align}
R &:=\sum_i\ket{i}\otimes P_i\quad\text{and} \label{R}\\
S &:=\sum_k\ket{k}\!\!\bra{k}\otimes U_k. \label{S}
\end{align}
$U_k$ is a unitary operator defined as
\begin{align}
U_k:=\sum_i\frac{1}{O_{i,j}O_{i+k,j}}P_i Q_j P_{i+k}. \label{U}
\end{align}
Note that the subscripts in both $P_{i+k}$ and $O_{i+k,j}$ are taken as sum modulo $d$, where the result is in the set $\{1,~...,~d\}$. Then, the action of $U_k$ produces a cyclic shift on $\{P_i\}_{i=1}^d$:
\begin{align}
U_kP_iU_k^\dagger=P_{i-k}.
\end{align}
Finally, $V_B$ is rewritten as
\begin{align}
V_B=\sum_i\ket{i}\otimes U_iP_i=\frac{1}{O_{d,j}}\sum_i\frac{1}{O_{i,j}}\ket{i}\otimes P_dQ_jP_i.
\end{align}

For Alice's side, let us first define a pair of operators whose action in $\mathcal{H}_A$ is analogous to the action of $P_{x_1}$ and $Q_{x_2}$ on Bob's side:
\begin{align}
\begin{aligned}
\tilde{P}_{x_1}&:=\frac{1}{d}\bigg(\mathds{1}+\sum_{x_2}\lambda_xA_x\bigg),\\
\tilde{Q}_{x_2}&:=\frac{1}{d}\bigg(\mathds{1}-\sum_{x_1}\lambda_xA_x\bigg).
\end{aligned}\qquad\forall\; x=x_1x_2
\end{align}
If written as in the above equation, $\tilde{P}_{x_1}$ and $\tilde{Q}_{x_2}$ satisfy
\begin{align}
\begin{aligned}
(\tilde{P}_{x_1}\otimes\mathds{1})\rho_{AB} &= (\mathds{1}\otimes P_{x_1})\rho_{AB}, \\
(\tilde{Q}_{x_2}\otimes\mathds{1})\rho_{AB} &= (\mathds{1}\otimes Q_{x_2})\rho_{AB}, 
\end{aligned} \label{cross relations}
\end{align}
for all $\rho_{AB}$, $\{P_{x_1}\}_{x_1=1}^d$ and $\{Q_{x_2}\}_{x_2=1}^d$ that saturate the quantum value. The new operators $\tilde{P}_{x_1}$ and $\tilde{Q}_{x_2}$ satisfy similar algebraic relations that their counterparts on Bob's side. By taking the partial trace over Bob's marginal state, we can easily verify from Eqs.~\eqref{cross relations} that
\begin{align}
    \begin{aligned}
        \tilde{P}_{x_1}\rho_A &= \tilde{P}_{x_1}^2\rho_A, \\
        \tilde{Q}_{x_2}\rho_A &= \tilde{Q}_{x_2}^2\rho_A.
    \end{aligned}
\end{align}
Also, an analogous version of Eqs.~\eqref{equation1} and \eqref{equation2} must be satisfied by $\tilde{P}_{x_1}$ and $\tilde{Q}_{x_2}$:
\begin{align}
(\tilde{P}_{x_1}\otimes\mathds{1})\rho_{AB} &= (\mathds{1}\otimes P_{x_1})\rho_{AB}=\frac{1}{O_{x_1,x_2}^2}(\mathds{1}\otimes P_{x_1}Q_{x_2}P_{x_1})\rho_{AB}=\frac{1}{O_{x_1,x_2}^2}(\tilde{P}_{x_1}\tilde{Q}_{x_2}\tilde{P}_{x_1}\otimes\mathds{1})\rho_{AB},
\end{align}
which leads to
\begin{align}
\begin{aligned}
O_{x_1,x_2}^2\tilde{P}_{x_1}&= \tilde{P}_{x_1}\tilde{Q}_{x_2}\tilde{P}_{x_1}\quad\text{and}\\
O_{x_1,x_2}^2\tilde{Q}_{x_1}&= \tilde{Q}_{x_2}\tilde{P}_{x_1}\tilde{Q}_{x_2}.
\end{aligned}
\end{align}
Lastly, defining $\tilde{R}$, $\tilde{S}$ and $\tilde{U}_k$ analogously to Eqs.~\eqref{R}, \eqref{S} and \eqref{U}, respectively, by replacing $P_{x_1}\rightarrow \tilde{P}_{x_1}$ and $Q_{x_2}\rightarrow \tilde{Q}_{x_2}$, we get a similar expression for $V_A$:
\begin{align}
V_A:= \sum_i \ket{i}\otimes \tilde{U}_i\tilde{P}_i.
\end{align}

Now, let us evaluate the action of $V$ on the state. Defining $\rho_\text{out}: = (V_A\otimes V_B)\rho_{AB}(V_A^\dagger\otimes V_B^\dagger)$ gives
\begin{align}
\rho_\text{out} = \sum_{i,i^\prime}\sum_{j,j^\prime}\ket{i,i^\prime}\!\!\bra{j,j^\prime}\otimes (\tilde{U}_i\tilde{P}_i\otimes U_{i^\prime}P_{i^\prime})\rho_{AB}(\tilde{P}_j\tilde{U}_j^\dagger\otimes P_{j^\prime}U_{j^\prime}^\dagger)
\end{align}
and because of Eqs.~\eqref{cross relations},
\begin{align}
(\tilde{U}_i\tilde{P}_i\otimes U_{i^\prime}P_{i^\prime})\rho_{AB}(\tilde{P}_j\tilde{U}_j^\dagger\otimes P_{j^\prime}U_{j^\prime}^\dagger) &= (\tilde{U}_i\otimes U_{i^\prime}P_{i^\prime})(\tilde{P}_i\otimes\mathds{1})\rho_{AB}(\tilde{P}_j\otimes\mathds{1})(\tilde{U}_j^\dagger\otimes P_{j^\prime}U_{j^\prime}^\dagger) \nonumber\\
&= (\tilde{U}_i\otimes U_{i^\prime}P_{i^\prime})(\mathds{1}\otimes P_{i})\rho_{AB}(\mathds{1}\otimes P_{j})(\tilde{U}_j^\dagger\otimes P_{j^\prime}U_{j^\prime}^\dagger) \nonumber\\
&= \delta_{i,i^\prime}\delta_{j,j^\prime}(\tilde{U}_i\otimes U_{i}P_{i})\rho_{AB}(\tilde{U}_j^\dagger\otimes P_{j}U_{j}^\dagger).
\end{align}
Thus, $\rho_\text{out}$ turns to be
\begin{align}
\rho_\text{out} = \sum_{i,j}\ket{i,i}\!\!\bra{j,j} \otimes (\tilde{U}_i\otimes U_{i}P_{i})\rho_{AB}(\tilde{U}_j^\dagger\otimes P_{j}U_{j}^\dagger).
\end{align}
Using Eqs.~\eqref{cross relations} it is possible to show that
\begin{align}
(\tilde{U}_k\otimes \mathds{1})\rho_{AB}=(\mathds{1}\otimes U_k^\dagger)\rho_{AB},
\end{align}
and
\begin{align}
\rho_\text{out} = \sum_{i,j}\ket{i,i}\!\!\bra{j,j} \otimes (\mathds{1}\otimes U_{i}P_{i}U_i^\dagger)\rho_{AB}(\mathds{1}\otimes U_{j}P_{j}U_{j}^\dagger) = \Phi^+_d \otimes d(\mathds{1}\otimes P_d)\rho_{AB}(\mathds{1}\otimes P_d). \label{two conclusions}
\end{align}
Finally, as $\rho_\text{out}$ must constitute a normalized state, we must have
\begin{align}
\tr \left[\left(\mathds{1}\otimes P_d\right)\rho_{AB}\right]=\frac{1}{d}.
\end{align}
Eq.~\eqref{two conclusions} leads us to two conclusions. Firstly, it is possible to certify the shared state whenever $O$ is a matrix possessing at least one column filled with non-zero elements, as in Eq.~\eqref{U}. Secondly, as $P_d$ is arbitrarily labeled, the maximal violation also implies that the marginal distribution of $\{P_{x_1}\}_{x_1=1}^d$ is uniform. In fact, replacing $U_k$ in Eq.~\eqref{S} by $U_{k+l}$, where $l\in\{1,~...,~d\}$, leads to
\begin{align}
    \tr \left[\left(\mathds{1}\otimes P_{x_1}\right)\rho_{AB}\right]=\frac{1}{d}\quad\forall\: x_1\in[\, 1,\, ...,\, d\, ].
\end{align}
In addition, note that $P_i$ and $Q_j$ in Eqs.~\eqref{R} and \eqref{U} can be interchanged. If one does so, the state certification can be made whenever $O$ has at least one non-zero row, and, in this case, the marginal distribution of $\{Q_{x_2}\}_{x_2=1}^d$ is guaranteed to be uniform:
\begin{align}
    \tr \left[\left(\mathds{1}\otimes Q_{x_2}\right)\rho_{AB}\right]=\frac{1}{d}\quad\forall\: x_2\in[\, 1,\, ...,\, d\, ].
\end{align}

\subsection{Exceptional cases for state certification} \label{exceptional cases}

\noindent The presented method for certification works for a wide class of overlap matrices, but it still requires that $O$ must have at least one column or row whose entries are non-zero. This case does not include, for instance, overlap matrices as in the Eq.~\eqref{even matrix}. In this section, let us provide an example of a state different than $\ket{\Phi_d^+}$ that also achieves $\beta_Q$ for such cases.

To introduce it, suppose a realization of $\mathcal{F}_4$ characterized by $P_{x_1}=\ket{x_1}\!\!\bra{x_1}$ and $Q_{x_2}=\ket{f_{x_2}}\!\!\bra{f_{x_2}}$, where $x_1,\, x_2=1,\, ...,\, 4$, and
\begin{align}
    \begin{aligned}
    \ket{f_{1}} = \frac{1}{\sqrt{2}}\left(\ket{1}+\ket{2}\right), &\quad 
    \ket{f_{2}} = \frac{1}{\sqrt{2}}\left(\ket{1}-\ket{2}\right), \\
    \ket{f_{3}} = \frac{1}{\sqrt{2}}\left(\ket{3}+\ket{4}\right), &\;\:\text{and}\;\:
    \ket{f_{4}} = \frac{1}{\sqrt{2}}\left(\ket{3}-\ket{4}\right).
    \end{aligned}
\end{align}
If considered so, the overlap matrix is given by
\begin{align}
    O=\begin{bmatrix}
    \sfrac{1}{\sqrt{2}} & \sfrac{1}{\sqrt{2}} & 0 & 0 \\
    \sfrac{1}{\sqrt{2}} & \sfrac{1}{\sqrt{2}} & 0 & 0 \\
    0 & 0 & \sfrac{1}{\sqrt{2}} & \sfrac{1}{\sqrt{2}} \\
    0 & 0 & \sfrac{1}{\sqrt{2}} & \sfrac{1}{\sqrt{2}}
    \end{bmatrix}, \label{example}
\end{align}
that is, it has the form of Eq.~\eqref{even matrix}, for $d=4$. In the same way as before, the observables of Alice are defined as in Eq.~\eqref{alicesmeas}. On the other hand, the state is given by
\begin{align}
    \rho_{AB}=\sum_{k=1}^2\left(\mathds{1}\otimes\Pi_k\right)\Phi_4^+\left(\mathds{1}\otimes\Pi_k\right)
\end{align}
where $\Phi_4^+$ is the four-dimensional maximally entangled state and
\begin{align}
    \begin{aligned}
        \Pi_1 &=\ket{1}\!\!\bra{1}+\ket{2}\!\!\bra{2} \\
        \Pi_2 &=\ket{3}\!\!\bra{3}+\ket{4}\!\!\bra{4}.
    \end{aligned}
\end{align}
If written in this way, it is easy to see that
\begin{align}
    \Pi_1P_{x_1}\Pi_1 = \begin{cases}
        P_{x_1} &\text{if}~x_1=1,\, 2, \\
        0 &\text{if}~x_1=3,\, 4
    \end{cases}
    \quad\text{and}\quad
    \Pi_2P_{x_1}\Pi_2 = \begin{cases}
        0 &\text{if}~x_1=1,\, 2, \\
        P_{x_1} &\text{if}~x_1=3,\, 4,
    \end{cases}
\end{align}
which can be generalized to both measurements of Bob and expressed in short as
\begin{align}
    \sum_{k=1}^2 \Pi_k P_{x_1}\Pi_k = P_{x_1}\;\;\text{and}\;\;\sum_{k=1}^2 \Pi_k Q_{x_2}\Pi_k = Q_{x_2}. \label{projectors on P}
\end{align}

Now, note that this realization not only saturates the quantum value of $\mathcal{F}_4$, but preserves the same statistics that would be produced by substituting the four-dimensional maximally entangled state. To check it, suppose an arbitrary observable $\Gamma$ of Alice, then
\begin{align}
    \begin{aligned}
        \tr\left[\left(\Gamma\otimes P_{x_1}\right)\Phi_4^+\right] &= \sum_k \tr\left[(\mathds{1}\otimes \Pi_k)(\Gamma\otimes P_{x_1})(\mathds{1}\otimes \Pi_k)\Phi_4^+\right] \\ 
        &= \tr\left[(\Gamma\otimes P_{x_1})\sum_k (\mathds{1}\otimes \Pi_k)\Phi_4^+(\mathds{1}\otimes \Pi_k)\right] = \tr\left[\left(\Gamma\otimes P_{x_1}\right)\rho_{AB}\right],
    \end{aligned}
\end{align}
where we just used Eq.~\eqref{projectors on P} for $P_{x_1}$ and the cyclicity of the trace.

To verify that there is a clear difference between $\rho_{AB}$ and $\Phi_4^+$, consider an isometry $V:=V_A\otimes V_B$ such that the action of $V_A$ is defined as
\begin{align}
    \begin{aligned}
        V_A\ket{1}_A &= \ket{1}_A\ket{1}_{A^\prime}, \\
        V_A\ket{2}_A &= \ket{1}_A\ket{2}_{A^\prime}, \\
        V_A\ket{3}_A &= \ket{2}_A\ket{1}_{A^\prime}\quad\text{and} \\
        V_A\ket{4}_A &= \ket{2}_A\ket{2}_{A^\prime},
    \end{aligned}
\end{align}
where $\ket{\cdot}_A$ denotes the marginal state of Alice, and $\ket{\cdot}_{A^\prime}$ denotes the marginal state of an ancillary subsystem of Alice. The action of $V_B$ is defined in the same way for the marginal state of Bob and its ancillary subsystem. Then, by applying such an isometry to $\rho_{AB}$ one gets
\begin{align}
    V\rho_{AB}V^\dagger=\frac{1}{2}\left(\ket{1}\!\!\bra{1}_A\otimes\ket{1}\!\!\bra{1}_B+\ket{2}\!\!\bra{2}_A\otimes\ket{2}\!\!\bra{2}_B\right)\otimes\left(\Phi_2^+\right)_{A^\prime B^\prime},
\end{align}
where $\left(\Phi_2^+\right)_{A^\prime B^\prime}$ is the two-dimensional maximally entangled state between subsystems $A^\prime$ and $B^\prime$. In other words, one can show that, up to local isometries, $\rho_{AB}$ is equivalent to a perfectly correlated classical random bit combined with a single copy of $\left(\Phi_2^+\right)_{A^\prime B^\prime}$.

One can extend this simple example to the case where the matrix of overlaps has a block-diagonal structure with $K$ blocks of dimension $d_k$, where $k=1,~...,~K$. In this case, consider the state
\begin{align}
    \rho_{AB}=\sum_{k=1}^K \left(\mathds{1}\otimes\Pi_k \right)\Phi_d^+ \left(\mathds{1}\otimes \Pi_k\right),
\end{align}
where $\Pi_k$ projects into the $d_k$-dimensional subspace of the $k$-th block of $O$. Using a proper definition of $V$, one can show that
\begin{align}
    V\rho_{AB}V^\dagger = \sum_{k=1}^K p_k\ket{k}\!\!\bra{k}_A\otimes \ket{k}\!\!\bra{k}_B\otimes\left(\Phi^+_{d_k}\right)_{A^\prime B^\prime},
\end{align}
where $p_k=d_k/d$ and $\left(\Phi^+_{d_k}\right)_{A^\prime B^\prime}$ is the $d_k$-dimensional maximally entangled state between subsystems $A^\prime$ and $B^\prime$. That is, it is possible to show that, up to local isometries, $\rho_{AB}$ is equivalent to a convex combination of maximally entangled states of various dimensions, where the classical registers tell Alice and Bob which state they share.

\section{The local value of \texorpdfstring{$\mathcal{F}_d$}{Lg}} \label{proofs for sec. 4}

\noindent In this appendix, we proceed as follows: first, we obtain a universal lower bound for the local value. This constitutes a proof for Theorem \ref{lower bound theorem}. Next, we verify the existence of a $d$-dimensional overlap matrix that achieves this lower bound only for the even $d$ cases. For odd $d$, we show that there is no such a matrix.

\subsection{Preliminaries for Theorem \ref{lower bound theorem}} \label{preliminaries}

\noindent Before starting with the demonstration of Theorem \ref{lower bound theorem}, some short auxiliary results are required. This section develops the solutions to two minimization problems arising while in the demonstration.

\subsubsection{The constrained probability simplex} 

\noindent In the proof of Theorem \ref{lower bound theorem}, we are required to minimize a concave function over a polytope. In this subsection, let us characterize the extremal points of this polytope. We start by representing the convex set by variables $\{t_i\}_{i=1}^{n}$ such that $t_i\ge 0$, $\sum_i t_i = 1$ and $t_i \le \tau$ for a fixed $\tau\in(0,\, 1)$. The first two constraints simply give us the probability simplex. The last one can be interpreted as hyperplanes that cut off the vertices of this simplex. Therefore, we refer to the resulting set as the constrained probability simplex. 

Let the $n$-tuple $\mathbf{t}=(t_1,\, t_2,\, ...,\, t_{n})$ represent a point inside of the probability simplex. Its extremal points are given by the deterministic distributions i.e., the permutations of
\begin{align}
\mathbf{t}=(1,\, 0,\, 0,\, ...,\, 0). \label{simplex extremal}
\end{align}
Naturally, if we require that $t_i\le \tau$, these vertices no longer belong to the set. Then, let us show that the new extremal points of the constrained probability simplex are given by permutations of
\begin{align}
\mathbf{t}=\bigg(\underbrace{\tau,\, \tau,\, ...,\, \tau}_{\times \left\lfloor \frac{1}{\tau}\right\rfloor},\, 1-\left\lfloor \frac{1}{\tau}\right\rfloor \times \tau,\, 0,\, 0,\, ...,\, 0\bigg). \label{extremal points}
\end{align}

To see that there are no other extremal points, one can wonder how many coordinates are admitted to take a value other than zero or $\tau$. Let us suppose that $t_j$ and $t_{j+1}$ are two coordinates satisfying this condition, that is, $0<t_j,\: t_{j+1}<\tau$. An arbitrary point $\mathbf{t}_\text{arb}$ with components $t_j$ and $t_{j+1}$ must be, up to permutations,
\begin{align}
\mathbf{t}_\text{arb}=\left(\tau,\, \tau,\, ...,\, \tau,\, t_j,\, t_{j+1},\, 0,\, 0,\, ...,\, 0\right).
\end{align}
If $\mathbf{t}_\text{arb}$ is not extremal, it is possible to decompose it into a convex sum. Changing this distribution for a quantity $\epsilon$, we get
\begin{align}
\mathbf{t}_+=\left(\tau,\, \tau,\, ...,\, \tau,\, t_j+\epsilon,\, t_{j+1}-\epsilon,\, 0,\, 0,\, ...,\, 0\right)
\end{align}
or, in the inverse way,
\begin{align}
\mathbf{t}_-=\left(\tau,\, \tau,\, ...,\, \tau,\, t_j-\epsilon,\, t_{j+1}+\epsilon,\, 0,\, 0,\, ...,\, 0\right).
\end{align}
Note that, for a small enough $\epsilon$, $\mathbf{t}_-$ and $\mathbf{t}_+$ are still valid points. However, $\mathbf{t}_\text{arb}=\sfrac{1}{2}\left(\mathbf{t}_++\mathbf{t}_-\right)$, so $\mathbf{t}_\text{arb}$ is not extremal. Thus, any extremal point must have zero or one coordinate in the open interval $0<t_i<\tau$. The first case appears only when $\sfrac{1}{\tau}$ is an integer. Otherwise, we find ourselves in the second case. In either case, the extremal points must have the form of the vector $\mathbf{t}$ given in Eq.~\eqref{extremal points}, up to permutations.

Having shown that all the extremal points have the form given above, let us argue that all those points are in fact extremal. To do so, it suffices to show that no single point from $\mathbf{t}$ can be written as a convex combination of the remaining points. More specifically, suppose that $\mathbf{t}_j$, with $j=1,~...,~J$, are points of the form of $\mathbf{t}$, in Eq.~\eqref{extremal points}. Then, let us show that those points, in particular $\mathbf{t}_1$, cannot be decomposed into a convex sum of $\mathbf{t}_j$, for $j\neq 1$. Without loss of generality, assume that, for some strictly positive weights $w_j$, where $\sum_j w_j=1$, $\mathbf{t}_1$ can be written as a convex sum:
\begin{align}
    \mathbf{t}_1=\sum_{j\in\mathcal{J}} w_j\mathbf{t}_j,
\end{align}
where $\mathcal{J}$ is the set of indexes $j$ for which $w_j>0$. Then, suppose that the first component of $\mathbf{t}_1$ is $[\mathbf{t}_1]_1=\tau$, so
\begin{align}
    [\mathbf{t}_1]_1=\sum_{j\in\mathcal{J}} w_j[\mathbf{t}_j]_1=\tau.
\end{align}
By hypothesis, $[\mathbf{t}_j]_1\le \tau$, for all $j$, and
\begin{align}
    [\mathbf{t}_1]_1=\sum_{j\in\mathcal{J}} w_{j}[\mathbf{t}_{j}]_1\le \tau \sum_{j\in\mathcal{J}} w_{j}= \tau. \label{saturation argument}
\end{align}
In other words, Ineq.~\eqref{saturation argument} is saturated, so $[\mathbf{t}_{j}]= \tau$, for all $j\in\mathcal{J}$. We can repeat this same argument to all of the components of $\mathbf{t}_1$ equal to $\tau$, leading to the same conclusion. In addition, using that the components of $\mathbf{t}_j$ are non-negative, a similar argument can be used to show that if $[\mathbf{t}_1]_k=0$, then $[\mathbf{t}_{j}]_k=0$, for all $j\in\mathcal{J}$. The last component of $\mathbf{t}_1$ equals $1-\left\lfloor\frac{1}{\tau}\right\rfloor\times\tau$, which is fixed by normalization, and so are those of $\mathbf{t}_{j}$, for all $j\in\mathcal{J}$. Therefore, $\mathbf{t}_1$ cannot be decomposed into a convex sum of $t_j$, for $j\neq 1$. As $\mathbf{t}_1$ is an arbitrary point among all of the $\mathbf{t}_j$ points, for all $j$, then all of the points of the form of $\mathbf{t}$, in Eq.~\eqref{extremal points}, are extremal.

\subsubsection{Minimization of a specific function}

\noindent In the proof of Theorem 3, we are required to find the minimum value of the following function:
\begin{align}
s(\tau): = 2\left\lfloor\frac{1-\tau}{\tau}\right\rfloor\left(\sqrt{1-\tau}-1\right)+2\sqrt{\tau\left(1+\left\lfloor\frac{1-\tau}{\tau}\right\rfloor\right)}-\tau-2 \label{the function}
\end{align}
over $\tau\in\left(0,\, 1\right]$. To do so, let us first show that $s(\tau)$ is continuous. By looking at the function in Eq.~\eqref{the function}, it can be noted that the possible discontinuous points are those at which $(1-\tau)/\tau$ is an integer. So, let us evaluate the sided limits of $s(\tau)$ at points of the form of $\tau=\sfrac{1}{n}$, where $n$ is a positive integer:
\begin{align}
\lim_{\epsilon\rightarrow 0^+} \; s\left(\sfrac{1}{n}-\epsilon\right)
&= 2(n-1)\left(\sqrt{1-\sfrac{1}{n}}-1\right)-\sfrac{1}{n}
\end{align}
and
\begin{align}
\lim_{\epsilon\rightarrow 0^+} \; s\left(\sfrac{1}{n}+\epsilon\right)
&= 2(n-2)\left(\sqrt{1-\sfrac{1}{n}}-1\right)+2\left(\sqrt{1-\sfrac{1}{n}}-1\right)-\sfrac{1}{n} \nonumber \\
&= 2(n-1)\left(\sqrt{1-\sfrac{1}{n}}-1\right)-\sfrac{1}{n}.
\end{align}
Since the sided limits coincide, $s(\tau)$ is continuous w.r.t. $\tau$.

Now, note that, inside of the interval $\tau\in\left[\frac{1}{n+1},\,\frac{1}{n}\right]$, $s(\tau)$ can be written as
\begin{align}
s(\tau)=2(n-1)(\sqrt{1-\tau}-1)+2\sqrt{n\tau}-\tau-2.
\end{align}
If we treat $n$ as a fixed parameter, it is easy to see that $s(\tau)$ is a concave function of $\tau$ (a linear combination of concave terms with non-negative coefficients is concave). Moreover, for a single-variable concave function defined over a closed and bounded interval, the points that minimize this function are at the edges of the interval. In this case, we are looking for points of the form $\tau=\sfrac{1}{n}$. Therefore, to minimize $s(\tau)$, we can discard the points in the interior and focus only at the edges.

Then, let us define a function $g(n)$ whose domain is the set of positive integers:
\begin{align}
g(n): =s\left(\tau=\frac{1}{n}\right)=2(n-1)\left(\sqrt{1-\frac{1}{n}}-1\right)-\frac{1}{n}.
\end{align}
For a moment, assume that $n$ is a continuous variable. This way, $g(n)$ can be minimized just by evaluating its derivative:
\begin{align}
\frac{dg}{dn}=2\left(\sqrt{1-\frac{1}{n}}-1\right)+\frac{n-1}{n^2\sqrt{1-\sfrac{1}{n}}}+\frac{1}{n^2}.
\end{align}
Setting $\frac{dg}{dn}=0$ after some algebraic manipulation leads to:
\begin{align}
n^{*2}-n^*-1=0\quad\Rightarrow\;\; n^*=\frac{1+\sqrt{5}}{2}\approx 1.62,
\end{align}
which provides a single positive root for $n^*$. Moreover, the second derivative, $\frac{d^2g}{dn^2}$, is positive at $n = n^*$ and so $n^*$ is indeed a minimum of $g(n)$. However, $g(n)$ is only defined for positive integers, and we have to restrict the analysis to this set. Since $n^*$ is the only critical point of $g(n)$ and $\frac{dg}{dn}$ is positive for $n>n^*$ and negative for $n<n^*$, it suffices to check the closest integers of $n^*$, i.e., $n=1$ and 2.

Then, returning to $s(\tau)$, as $s(\tau=1)=-1$ and $s\left(\tau=\sfrac{1}{2}\right)=\sqrt{2}-\sfrac{5}{2}$, $s(\tau)$ is minimized at $\tau=\sfrac{1}{2}$.

\subsection{The analysis of the local value in Eq.~(\ref{local value})} \label{proof of the lower bound theorem}

\noindent In this section, let us show some properties that can be extracted from the expression of the local value in Eq.~\eqref{local value}. Firstly, suppose that $O_{ij}=1$, for some $i,~j\in\{1,~...,~d\}$. Note that this assumption leads to $O_{i,x_2}=\delta_{x_2,j}$ and $O_{x_1,j}=\delta_{x_1,i}$, as the rows and columns of $O$ are normalized. Next, Eq.~\eqref{local bound 2} evaluated for the strategy $s(i,\, j)$ of Bob leads to
\begin{align}
s(i,\, j) &= \sum_{x_1\neq i}^{d} \left[\sqrt{1-O_{x_1,j}^2}-\frac{1}{2}\left(1-O_{x_1,j}^2\right)\right]+\sum_{x_2\neq j}^{d} \left[\sqrt{1-O_{i,x_2}^2}-\frac{1}{2}\left(1-O_{i,x_2}^2\right)\right] = d-1.
\end{align}
Note that this is the largest value that can be achieved by any strategy, as the quantum value is also $d-1$. Therefore, in this case, this must be the optimal strategy of Bob, and $\beta_L(O)=d-1$. This is the reason why the assumption $O_{x_1x_2}<1$ is required for most of our analysis. If one desires to trivialize the Bell inequalities presented in this work, it suffices to construct $\mathcal{F}_d$ out of a matrix with a single overlap equal to 1.

Now, let us show that $\beta_L(O)$ can be lower bounded by a clever choice of the strategies of Bob. This constitutes a proof for Theorem 3. Without loss of generality, let us start by identifying the largest element of $O$ as $O_{11}$. We can always do so, as it is always possible to relabel the outputs. Then, let us lower bound $\beta_L$ by restricting the set of strategies from $u,\, v\in\{1,~...,~d\}$ to $u,\,v\in\{1\}$:
\begin{align}
\beta_L(O)=\underset{u,v}{\text{max}}\; \left[\, s(u, v)\,\right] \ge \underset{u,v\in \{1\}}{\text{max}}\; \left[\, s(u, v)\,\right] =  s\left(1,\, 1\right).
\end{align}
In other words, we choose to bound $\beta_L$ by the strategy related to the largest element of $O$, which we defined to be $O_{11}$. This particular choice circumvents the maximization and reduces the problem of lower bounding $\beta_L$ to the calculation of a single strategy. Recall that, from Eq.~\eqref{local bound 2}, we are summing over the interval $R_{\pm}=\{x\in [d]^2\, |\, (\delta_{x_1u}-\delta_{x_2v})=\pm 1\}$, which means that the sum is performed over row $u$ and column $v$, but it excludes the term $(u,~v)$, as $(\delta_{uu}-\delta_{vv})$ is, obviously, zero. Then, by making this choice, we are lower bounding $\beta_L$ by the strategy that excludes the largest term of $O$ in the sum. 

The strategy $s(1,\, 1)$, in turn, is given by
\begin{align}
s(1,\, 1) &= \sum_{x_1\neq 1}^d \left[\sqrt{1-O_{x_1,1}^2}-\frac{1}{2}\left(1-O_{x_1,1}^2\right)\right]+\sum_{x_2\neq 1}^d \left[\sqrt{1-O_{1,x_2}^2}-\frac{1}{2}\left(1-O_{1,x_2}^2\right)\right].
\end{align}
Defining $t_{x_1x_2}: = O_{x_1x_2}^2$ and using the normalization of rows and columns of $O$, we get
\begin{align}
s(1,\, 1) = \sum_{x_1\neq 1}^d \sqrt{1-t_{x_1,1}}+\sum_{x_2\neq 1}^d \sqrt{1-t_{1,x_2}}-d-t_{11}+2. \label{minimization}
\end{align}
To continue, let us define an auxiliary function, $h$, such that
\begin{align}
h(t_{21},\, t_{31},\, ...,\, t_{d,1}) &:= \sum_{x_1=2}^d \sqrt{1-t_{x_1,1}}, \label{h function}
\end{align}
i.e., it corresponds to the first term of Eq.~\eqref{minimization}. Note that $h$ is a Schur-concave function w.r.t. the $d-1$ variables $(t_{21},\, t_{31},\, ...,\, t_{d,1})$. This is of particular importance because any Schur-concave function, when minimised over a compact set, achieves its minimum on some extremal point of this set. Thus, we can use this fact to minimize both functions.

Then, let us express the $d-1$ variables of $h$ as a vector, $\mathbf{t}:=(t_{21},\, t_{31},\, ...,\, t_{d,1})$. For the elements of $\mathbf{t}$, we have
\begin{align}
t_{x_1,1}\ge 0,\quad\sum_{x_1\neq 1}t_{x_1,1}=1-t_{11}\quad\text{and}\quad t_{x_1,1}\le t_{11}. \label{constrained simplex}
\end{align}
The last condition comes from the fact that $O_{11}$ (and, by consequence, $t_{11}$) is the largest element. The solution to the minimization of $h(\mathbf{t})$ is presented in Section \ref{preliminaries} where the extremal points $\mathbf{t}_\text{ext}$ of the set in Eqs.~\eqref{constrained simplex} are given by permutations of
\begin{align}
\mathbf{t}_\text{ext}=\bigg(\underbrace{t_{11},\, ...,\, t_{11}}_{\times \left\lfloor\frac{1-t_{11}}{t_{11}}\right\rfloor},1-t_{11}-\left\lfloor\frac{1-t_{11}}{t_{11}}\right\rfloor\times t_{11},0,\, ...,\, 0\bigg). \label{solution vector}
\end{align}
An identical solution is obtained for the minimization of the second term of Eq.~\eqref{minimization}, so $s(1,\, 1)$ can be lower bounded by
\begin{align}
s(1,\, 1) \ge 2\left\lfloor\frac{1-t_{11}}{t_{11}}\right\rfloor\left(\sqrt{1-t_{11}}-1\right)+2\sqrt{t_{11}\left(1+\left\lfloor\frac{1-t_{11}}{t_{11}}\right\rfloor\right)}-t_{11}+d-2. \label{solution}
\end{align}
Now, we have a single-variable function, which can be easily minimized. This is precisely the function analyzed in Section \ref{preliminaries} and the minimum is proven to occur at $t_{11}=\sfrac{1}{2}$. Thus,
\begin{align}
s(1,\, 1)\ge d+\sqrt{2}-\frac{5}{2},
\end{align}
which implies that $\beta_L(O) \ge d+\sqrt{2}-\sfrac{5}{2}$, for any $d$-dimensional $O$, with $d\ge 2$.

\subsection{Proof of Lemma \ref{optimal matrix}} \label{proof lemma 2}

\noindent In this section, we investigate the overlap matrices that saturate the lower bound on $\beta_L$. Let us start with the following lemma.

\begin{lemma} \label{smart corollary}
If a $d\times d$ matrix $O^*$ is such that $\beta_L(O^*)=d+\sqrt{2}-\sfrac{5}{2}$, then a $2\times 2$ MUB block can be separated from the remainder.
\end{lemma}

\begin{proof} Assume that $O^*$ is a $d\times d$ matrix such that $\beta_L(O^*)=d+\sqrt{2}-\sfrac{5}{2}$. From the argument developed in the last section, we start by assigning $\sfrac{1}{2}$ to $t_{11}$, or, in this case, assuming that $O^*_{11}=\sfrac{1}{\sqrt{2}}$. Because of the strict concavity of $h(\mathbf{t})$ in Eq.~\eqref{h function}, Eq.~\eqref{solution} is tight if and only if the vector that minimizes $h(\mathbf{t})$ is given by permutations of
\begin{align}
\mathbf{t}_\text{ext}=\left(\frac{1}{2},\, 0,\, ...,\, 0\right),
\end{align}
that is, to saturate the lower bound, we take $O^*_{12}=O^*_{21}=\sfrac{1}{\sqrt{2}}$. Besides that, we must also guarantee that the overlap matrix corresponds to a valid unitary. Recall that, from the Eq.~\eqref{overlap}, the overlaps are taken as the absolute value of the inner product between bases $\{e_{x_1}\}_{x_1=1}^d$ and $\{f_{x_2}\}_{x_2=1}^d$. In other words, the elements of $O$ are the absolute value of the elements of some unitary. It is easy to see that this implies that $O^*_{22}=\sfrac{1}{\sqrt{2}}$, so that the partial form of $O^*$ is given by
\begin{align}
O^*=\begin{bmatrix}
\sfrac{1}{\sqrt{2}} & \sfrac{1}{\sqrt{2}} & 0 & ... & 0\\
\sfrac{1}{\sqrt{2}} & \sfrac{1}{\sqrt{2}} & 0 & ... & 0\\
0 & 0 & & & \\
\vdots & \vdots & & O_\text{rest} & \\
0 & 0 & & & \\
\end{bmatrix}, \label{optimal matrix 0}
\end{align}
where $O_\text{rest}$ represents the uncharacterized elements of $O^*$. \end{proof}

Now, to continue with the proof of Lemma 2, we have to obtain some characterization of $\beta_L(O_\text{rest})$. In fact, we will show that
\begin{align}
\beta_L(O_\text{rest}) = (d-2)+\sqrt{2}+\frac{5}{2},
\end{align}
which coincides with the lower bound for dimension $d - 2$. 

Because $O_\text{rest}$ is a square matrix of dimension $d-2$, the Theorem \ref{lower bound theorem} already implies that $\beta_L(O_\text{rest})$ is lower bounded by $(d-2)+\sqrt{2}+\sfrac{5}{2}$. On the other hand, the upper bound is a consequence of Lemma \ref{smart corollary}. Because $\beta_L(O^*)=d+\sqrt{2}-\sfrac{5}{2}$, all strategies of $O^*$ must be upper bounded by this same amount. In particular, note that for each strategy labeled by outputs inside of the $O_\text{rest}$ block, in Eq.~\eqref{optimal matrix 0}, there are four zero terms, whose overall contribution is $+2$. Then, when considering only $O_\text{rest}$, we have $\beta_L(O_\text{rest})\le (d-2)+\sqrt{2}+\sfrac{5}{2}$.

Thus, for dimension $d-2$, the block $O_\text{rest}$ also fulfils the assumptions of Lemma \ref{smart corollary}, so we can also extract a $2\times 2$ MUB block from it. Naturally, the indefinite iteration of the above argument will lead us to two cases: either $d$ is even and
\begin{align}
O^* =\begin{bmatrix}
\sfrac{1}{\sqrt{2}} & \sfrac{1}{\sqrt{2}} & ... & 0 & 0\\
\sfrac{1}{\sqrt{2}} & \sfrac{1}{\sqrt{2}} & ... & 0 & 0\\
\vdots & \vdots & \ddots & \vdots & \vdots \\
0 & 0 & ... & \sfrac{1}{\sqrt{2}} & \sfrac{1}{\sqrt{2}} \\
0 & 0 & ... & \sfrac{1}{\sqrt{2}} & \sfrac{1}{\sqrt{2}} \\
\end{bmatrix}. \label{even matrix 1}
\end{align}
or $d$ is odd and
\begin{align}
O^*=\begin{bmatrix}
\sfrac{1}{\sqrt{2}} & \sfrac{1}{\sqrt{2}} & 0 & ... & 0\\
\sfrac{1}{\sqrt{2}} & \sfrac{1}{\sqrt{2}} & 0 & ... & 0\\
0 & 0 & & & \\
\vdots & \vdots & & \ddots \\
0 & 0 & & & 1\\
\end{bmatrix}, \label{odd matrix}
\end{align}
where the last block corresponds to the element $O^*_{d,d}=1$. However, while Eq.~\eqref{even matrix 1} is a correct form of $O^*$ in the even case, for odd $d$, this is not true. Firstly, because the matrix in Eq.~\eqref{odd matrix} has one overlap equal to one, which we discarded from our analysis. Secondly, we know, from Appendix \ref{proof of the lower bound theorem}, that this leads to $\beta_L(O^*)=d-1$, which contradicts the initial statement of Lemma \ref{smart corollary}, that $\beta_L(O^*)=d+\sqrt{2}-\sfrac{5}{2}$. Thus, for the odd $d$ case, there is no $O^*$ that saturates the lower bound in Theorem 3.

\subsection{Some considerations for the odd \texorpdfstring{$d$}{Lg} case} \label{explanation}

\noindent Finally, let us quickly explain why the lower bound derived in Theorem 3 is not tight for odd $d$. Consider the entry-wise squared version of $O$, which we refer to as $T$. If written so, $T$ assumes the form of a \emph{unistochastic} matrix i.e., if $U_{x_1x_2}$ are the elements of a unitary $U$, for $x_1,~x_2\in\{1,\, ...,\, d\}$, then $T_{x_1x_2}=|U_{x_1x_2}|^2$. A unistochastic matrix is also a particular case of a \emph{bistochastic} matrix -- a non-negative matrix whose columns and rows add up to one. For $2\times 2$ arrays, it happens that the bistochastic and the unistochastic sets of matrices coincide, but for larger dimensions the unistochastic set is a proper subset of the bistochastic set.

It is clear that our original intention was to minimize $\beta_{L}(O)$ over the set of overlap matrices. However, the proof of Theorem 3 relies only on the fact that $T$ is bistochastic and, in fact, the derived lower bound corresponds exactly to the lowest value achievable by a bistochastic matrix. For even $d$, among the optimal bistochastic matrices, there exist some which are also unistochastic and, hence, the resulting bound is tight.

For odd $d$, however, if we try to simultaneously saturate the bound and enforce unistochasticity, we reach a contradiction, as shown in the previous section. On the other hand, the lower bound can be saturated by a bistochastic matrix, e.g.:
\begin{align}
T^*=\begin{bmatrix}
\sfrac{1}{2} & 0 & 0 & ... & 0 & \sfrac{1}{2} \\
\sfrac{1}{2} & \sfrac{1}{2} & 0 & ... & 0 & 0 \\
0 & \sfrac{1}{2} & \sfrac{1}{2} & ... & 0 & 0 \\
\vdots & \vdots & \vdots & \ddots & \vdots & \vdots \\
0 & 0 & 0 & ... & \sfrac{1}{2} & \sfrac{1}{2} \\
\end{bmatrix}, \label{bistochastic optimal}
\end{align}
which is valid for all $d\ge 2$. However, none of the optimal bistochastic matrices for odd $d$ happen to be unistochastic.

Now, let us focus on the case $d=3$. An aspect that makes the optimization over the set of unistochastic matrices for $d = 3$ difficult is the fact that this set is not convex, but only star-convex. Also, the unistochastic and bistochastic sets are both centered at $\frac{1}{3}J_3$, where $J_3$ is the $3\times 3$ matrix of ones \cite{Bengtsson05}. A proper notion of center can be acquired if one considers a uniformly weighted convex combination of the extremal points of the permutation matrices. The permutation matrices are the only extremal points of the bistochastic set, and they are also extremal for the unistochastic set.

We have numerically implemented a function that calculates the local value for a given overlap matrix. By performing a local search over a large number of random starting points we have reached the conjecture that the smallest value of $\beta_L$ for $d = 3$ is achieved for
\begin{align}
O_\text{conj}=\begin{bmatrix}
\sfrac{1}{3} & \sfrac{2}{3} & \sfrac{2}{3} \\
\sfrac{2}{3} & \sfrac{1}{3} & \sfrac{2}{3} \\
\sfrac{2}{3} & \sfrac{2}{3} & \sfrac{1}{3} \\
\end{bmatrix}. \label{conjecture}
\end{align}
To see why this conjecture is reasonable, consider the entry-wise squared version of $O_\text{conj}$, which we refer to as $T_\text{conj}$. An analytic condition derived in Ref.~\cite{Fedullo92} allows us to check that $T_\text{conj}$ is not only unistochastic, but also lies at the boundary of the unistochastic set. Furthermore, consider a permutation of the three-dimensional matrix in Eq.~\eqref{bistochastic optimal},
\begin{align}
T^*_{3}=\begin{bmatrix}
0 & \sfrac{1}{2} & \sfrac{1}{2} \\
\sfrac{1}{2} & 0 & \sfrac{1}{2} \\
\sfrac{1}{2} & \sfrac{1}{2} & 0 \\
\end{bmatrix}.
\end{align}
Then,
\begin{align}
T_\text{conj}=\frac{1}{3}\left(\frac{1}{3}J_3\right)+\frac{2}{3}T^*_{3}.
\end{align}
That is, if one considers the line segment connecting the center of the unistochastic set and the optimal bistochastic matrix, $T_\text{conj}$ can be found at the intersection of this segment with the boundary of the unistochastic set.

Lastly, let us show that $O_\text{conj}$ provides a smaller local value than MUBs. By calculating $\beta_L(O_\text{conj})$, we get
\begin{align}
\beta_L(O_\text{conj}) = \frac{1}{9}\left[6\left(\sqrt{8}+\sqrt{5}\right)-13\right]\approx 1.9319.
\end{align}
For $d$-dimensional MUBs, we obtain
\begin{align}
\beta_L\left(\frac{1}{\sqrt{d}}J_d\right)= 2(d-1)\sqrt{\frac{d-1}{d}}-\frac{(d-1)^2}{d}.
\end{align}
If evaluated for $d = 3$, then $\beta_L\left(\sfrac{J_3}{\sqrt{3}}\right)=1.9327$, which is slightly bigger than $\beta_L(O_\text{conj})$. In fact, $O_\text{conj}$ allows us to construct counterexamples for all odd dimensions. For odd $d \geq 3$, consider the following matrix:
\begin{align}
O_{d}^\oplus=\left(\bigoplus_{i=1}^{\left\lfloor\sfrac{d}{2}\right\rfloor-1}\frac{1}{\sqrt{2}}J_2\right)\oplus O_\text{conj}. 
\end{align}
The local value obtained for $O_{d}^\oplus$ is given by
\begin{align}
\beta_L(O_{d}^\oplus) = d-3+\frac{1}{9}\left[6\left(\sqrt{8}+\sqrt{5}\right)-13\right]\approx d-1.0681.
\end{align}
If one takes the derivative of $\beta_L\left(\sfrac{J_d}{\sqrt{d}}\right) - \beta_L(O_{d}^\oplus)$, by a simple analytic argument it is possible to conclude that this derivative is positive for all odd $d \geq 3$. Because $\beta_L\left(\sfrac{J_d}{\sqrt{d}}\right) - \beta_L(O_{d}^\oplus)$ is positive for $d=3$, then it must also be positive for all odd $d\ge 3$. Therefore, this shows that for all odd $d\ge 3$, the realization of $\mathcal{F}_d$ that is most robust to noise does not correspond to MUBs in dimension $d$.

\end{document}